\documentclass[conference, 10pt]{IEEEtran} 
\usepackage{amsmath, amssymb, amsthm, bm, mathtools}
\usepackage[numbers]{natbib}
\usepackage{booktabs}
\newtheorem{theorem}{Theorem}
\newtheorem{corollary}{Corollary}

\newtheorem{definition}{Definition}

\newtheorem{remark}{Remark}
\usepackage{url}
\usepackage{multirow}
\usepackage{bbm}
\usepackage{graphicx} 
\usepackage[acronym]{glossaries}
\usepackage{tikz}
\usepackage{pgfplots}
\usepackage{algpseudocode, algorithm}
\pgfplotsset{compat=1.18}
\usetikzlibrary{shapes.multipart}
\usepackage{dsfont}

\usetikzlibrary{positioning,arrows.meta,fit}
\usetikzlibrary{shapes,arrows.meta}

\usepackage[caption=false]{subfig}

\usepackage{tikz}
\usetikzlibrary{positioning,arrows.meta,fit,calc,decorations.pathreplacing}

\usepackage{acronym}
\acrodef{3GPP}{3rd generation partnership project}
\acrodef{6G}{the sixth generation}
\acrodef{AI}{artificial intelligence}
\acrodef{AR}{autoregressive}

\acrodef{CDF}{cumulative distribution function}
\acrodef{CSI}{channel state information}
\acrodef{CQI}{channel quality indicator}
\acrodef{CNN}{convolutional neural network}
\acrodef{CSP}{communications service provider}
\acrodef{DNN}{deep neural network}
\acrodef{EMA}{exponential moving average}
\acrodef{GO-OAE}{goal-oriented ordered autoencoder}
\acrodef{GD}{gradient descent}
\acrodef{GUI}{graphical user interface}
\acrodef{IB}{information bottleneck}
\acrodef{KL}{Kullback-Leibler}
\acrodef{KKT}{Karush–Kuhn–Tucker}
\acrodef{MARL}{multi-agent reinforcement learning}
\acrodef{MADRL}{multi-agent deep reinforcement learning}
\acrodef{MDP}{Markov decision process}
\acrodef{ML}{machine learning}
\acrodef{MLP}{multi-layer perceptron}
\acrodef{MILP}{mixed integer linear programming}
\acrodef{MPIS}{multi-prefix importance sampling}
\acrodef{MRL}{Matryoshka representation learning}
\acrodef{NLP}{natural language processing}
\acrodef{OAE}{ordered autoencoder}
\acrodef{PMF}{probability mass function}
\acrodef{PSNR}{peak signal-to-noise ratio}
\acrodef{QoE}{quality of experience}
\acrodef{QoS}{quality of service}
\acrodef{SNR}{signal-to-noise ratio}
\acrodef{SGD}{stochastic gradient descent}
\acrodef{SINR}{signal-to-interference-plus-noise ratio}
\acrodef{SIR}{signal-to-interference ratio}
\acrodef{TCP}{transmission control protocol}
\acrodef{TX}{transmitter}
\acrodef{UE}{user equipment}
\acrodef{VAE}{variational autoencoder}
\acrodef{VQ}{vector quantization}
\acrodef{VQ-VAE}{vector quantized variational autoencoder}
\acrodef{VQ-AE}{vector quantized autoencoder}
\acrodef{VAE}{variational autoencoder}
\acrodef{ViT}{vision transformer}

\newcommand{\mc}[1]{\mathcal{#1}}   
\newcommand{\ve}[1]{\boldsymbol{\mathbf{#1}}} 
\newcommand{\field}[1]{\mathbb{#1}}
\newcommand{\operator}[1]{\mathrm{#1}}

\newcommand{\E}[1]{\mathbb{E}\left[ #1 \right]} 
\newcommand{\EE}{{\field{E}}} 
\newcommand{\R}{{\field{R}}} 

\newcommand{\Var}{{\operator{Var}}}

\newcommand{\Ds}{{\mc{D}}}
\newcommand{\Ks}{{\mc{K}}}
\newcommand{\Xs}{{\mc{X}}}
\newcommand{\Ys}{{\mc{Y}}}
\newcommand{\Zs}{{\mc{Z}}}
\newcommand{\Sss}{{\mc{S}}}
\newcommand{\Rs}{{\mc{R}}}

\newcommand{\Ls}{{\mc{L}}}
\newcommand{\Us}{{\mc{U}}}
\newcommand{\Cs}{{\mc{C}}}
\newcommand{\Os}{{\mc{O}}}
\newcommand{\Vs}{{\mc{V}}}
\newcommand{\Es}{{\mc{E}}}

\newcommand{\vx}{\ve{x}}
\newcommand{\vz}{\ve{z}}
\newcommand{\vy}{\ve{y}}

\newcommand{\vu}{\ve{u}}
\newcommand{\vh}{\ve{h}}

\newcommand{\vc}{\ve{c}}
\newcommand{\vnu}{\ve{\nu}}

\definecolor{mplblue}{HTML}{1f77b4}
\definecolor{mplorange}{HTML}{ff7f0e}
\definecolor{mplgreen}{HTML}{2ca02c}

\definecolor{UB}{HTML}{2c3e50}
\definecolor{RE_OAE}{HTML}{e74c3c}
\definecolor{RE_OAE_TA}{HTML}{e67e22}
\definecolor{GO_VQVAE}{HTML}{95a5a6}
\definecolor{GO_OAE_U_SC}{HTML}{3498db}
\definecolor{GO_OAE_U_I_SC}{HTML}{27ae60}
\definecolor{ViT}{HTML}{5c6bc0}

\definecolor{fig1E1}{HTML}{2c3e50}
\definecolor{fig1E2}{HTML}{e74c3c}
\definecolor{fig1E3}{HTML}{e67e22}
\definecolor{fig1E40}{HTML}{95a5a6}
\definecolor{fig1E40+supcon}{HTML}{7f8c8d}
\definecolor{fig1E4a}{HTML}{3498db}
\definecolor{fig1E4a+supcon}{HTML}{9b59b6}
\definecolor{fig1E4b}{HTML}{27ae60}

\pgfplotsset{compat=1.18}

\DeclareMathOperator*{\argmin}{arg\,min}    

\begin{document}
\title{
Goal-Oriented Ordered Autoencoding for Successive Refinement in Progressive Communication}

\author{
\IEEEauthorblockN{Qi Liao, Dario Bega}
\IEEEauthorblockA{ 
	Nokia Bell Labs, Stuttgart, Germany\\
E-Mail:
\url{{qi.liao, dario.bega}@nokia-bell-labs.com}} 
}
\maketitle

\begin{abstract} 
Goal-oriented communication often struggles to sustain downstream task performance under variable and unpredictable bandwidth. We address this problem by proposing GO-OAE (Goal-Oriented Ordered AutoEncoding), which learns a progressively refinable quantized latent sequence, optimized directly for task utility rather than source reconstruction. Each component of the sequence is represented for transmission by its code index, and the learned ordering enables inference from any received prefix of the index sequence, with performance improving as more indices arrive.
We prove that optimizing the goal-oriented distortion objective is equivalent to solving a weighted conditional information bottleneck formulation, and introduce an unbiased multi-prefix importance sampling (MPIS) estimator, to reduce GO-OAE training computational cost from $\mathcal{O}(BK)$ to $\mathcal{O}(Bm)$ with $m{\ll}K$, and reduces the estimator variance by a factor of $1/m$.
Experiments show that, on MiniImageNet, GO-OAE achieves over $96\%$ of the lossless upper-bound classification accuracy using $324$ bits per image ($683\times$ compression vs.\ raw RGB), and already attains over $94\%$ of the upper bound with only 70\% of the transmitted indices. It consistently outperforms reconstruction-optimized codecs and ordered autoencoders trained with uniform prefix sampling across all transmission budgets, demonstrating effective communication at extreme compression rates.
\end{abstract}


\section{Introduction}\label{sec:introduction}

Next-generation communication systems increasingly support machine intelligence rather than human perception, with transmitted data used for downstream tasks such as classification, detection, and decision-making. This shift calls for coding methods that preserve task-relevant information under bandwidth constraints, rather than merely minimizing pixel-level distortion. Although goal-oriented communication \cite{goldreich2012theory} captures this principle, it has yet to unify compression and inference while ensuring consistent downstream performance across varying and uncertain bandwidth constraints.

A common approach is to use generative neural codecs \cite{mentzer2020high, zhang2025progressive}, such as \ac{VAE} \cite{diederik2019introduction}, for compression followed by a separately trained task model. This modular ``compress first, infer later'' pipeline assumes that reconstruction fidelity preserves task-critical information. Yet our results in Section~\ref{sssec:reconstruction_fidelity} challenge this: even with good reconstruction quality, classification accuracy on reconstructed images remains far below the raw-image bound ($18.9\%$ vs. $77.5\%$), motivating \emph{goal-oriented codec-task co-design}.

Beyond the modular versus joint design question, practical communication systems can benefit from \emph{progressive transmission} \cite{equitz1991successive}, where an ordered latent sequence supports inference across varying bandwidths without re-encoding. The \ac{OAE} introduced by \cite{xu2021anytime} enables progressively improved reconstruction from increasingly long prefixes of the latent sequence. However, existing \acp{OAE} are tied to reconstruction and cannot recover task-relevant information. 

These observations expose a key research gap: 
existing frameworks lack an information-theoretic approach for learning emergent ordered discrete representations, whose prefixes are directly optimized for downstream tasks. Addressing this requires (i)~a goal-oriented objective that preserves progressive structure, (ii)~an efficient training method that avoids evaluating all prefix lengths, and (iii)~semantic regularization to maintain discriminative information in the codebook.
We address these challenges with \emph{\ac{GO-OAE}} (Figure~\ref{fig:go_oae_framework}). \ac{GO-OAE} jointly optimizes the encoder, vector quantizer, and task decoder through an unbiased importance-sampled multi-prefix objective. The resulting representations prioritize task-relevant information in earlier components without explicit architectural ordering constraints, enabling variable-bandwidth communication.


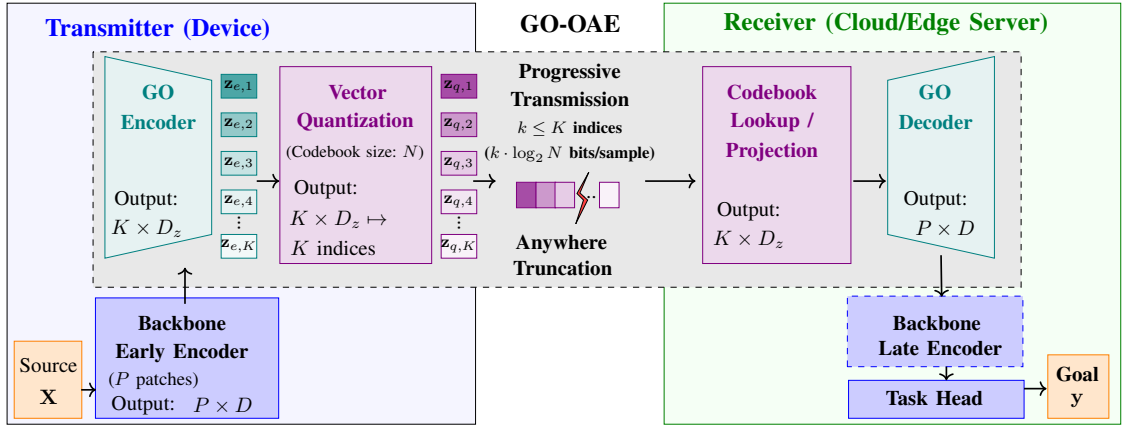
\begin{figure*}[t]
\begin{center}
\resizebox{.85\linewidth}{!}{
\begin{tikzpicture}
  \draw[green!50!black, fill=green!4] (11.13,-0.6) rectangle (18.83,6.5);
  \draw[fill=blue!4] (0.05,-0.6) rectangle (7.95,6.5);
  \draw[fill=black!9, dashed] (1.5,5.68) rectangle (17.15,1.71);
  \fill[violet!15, draw=violet] (11.78,2.11) rectangle (14.28,5.43);
  \draw[fill=teal!10, draw=teal] (1.7,5.6) -- (1.7,2.2) -- (3.5,2.6) -- (3.5,5.2) -- cycle;
  \draw[teal, fill=teal!10] (16.7,5.6) -- (16.7,2.2) -- (14.9,2.6) -- (14.9,5.2) -- cycle;
  \fill[violet!15, draw=violet] (4.65,2.11) rectangle (7.21,5.43);
  \fill[orange!20,draw=orange] (0.17,-0.5) rectangle (1.31,0.8);
  \fill[orange!20,draw=orange] (17.6,-0.5) rectangle (18.63,0.57);
  \fill[blue!20, draw=blue] (1.54,-0.5) rectangle (4.63,1.53);
  \node[draw=none, node font=\bfseries] at (3,1.1) {Backbone 
  };
  \node[draw=none, minimum width=106pt, node font=\bfseries] at (3,0.6) {Early Encoder 
      };
  \node[draw=none, node font=\bfseries, text=violet] at (5.96,5) {Vector  
    };
  \node[draw=none, node font=\bfseries, text=teal] (node2) at (2.6,5) {GO 
      };
  \node[draw=none, node font=\bfseries, text=violet] at (5.92,4.5) {Quantization 
      };
  \node[draw=none, node font=\bfseries, text=teal] at (2.6,4.5) {Encoder 
        };
  \node[draw=none, node font=\bfseries, text=teal] (node3) at (15.7,5) {GO 
        };
  \node[draw=none, node font=\bfseries, text=teal] at (15.71,4.5) {Decoder 
          };
  \node[draw=none] at (0.75,-0.14) { $\mathbf{X}$ 
      };
  \node[draw=none] at (0.75,0.4) {Source 
        };
  \node[draw=none, node font=\bfseries] (node7) at (18.1,-0.2) { $\mathbf{y}$ 
        };
  \node[draw=none, node font=\bfseries] at (18.1,0.24) {Goal 
          };
  \node[draw=none, node font=\bfseries] at (19.33,3.24) { 
            };
  \node[draw=none, text=black, node font=\bfseries] at (9.4,2.44) {Anywhere 

        };
  \node[draw=none, text=black, node font=\bfseries] at (9.54,5.32) {Progressive  };
  \node[draw=none, text=black, node font=\bfseries] at (9.54,4.88) {Transmission };
  \node[draw=none, node font=\footnotesize\bfseries] at (9.54,4.39) {$k\leq K$ indices };
  \node[draw=none, node font=\footnotesize\bfseries] at (9.54,3.95) {($k\cdot \log_2 N$ bits/sample) };
  \node[draw=none] at (2.35,-0.28) {Output: 
    };
  \node[draw=none] at (5.43,3.36) {Output: 
      };
  \node[draw=none] at (12.56, 3) {Output: 
        };
  \node[draw=none] at (3.69,-0.29) {$P\times D $ 
      };
  \node[draw=none] at (5.64,2.85) {$K \times D_z\mapsto$  
        };
  \node[draw=none] at (5.55,2.37) {$K$ indices  
          };
  \node[draw=none, node font=\footnotesize] at (5.92,3.97) {(Codebook size:  $N$) 
 
          };
  \node[draw=none, node font=\small] (node8) at (2.52,0.14) {($P$ patches)   
            };
  \node[draw=none, node font=\footnotesize\bfseries] at (3.8,3.03) { 
   
            };
  \node[draw=none] (node1) at (2.4,3.18) {Output: 
      };
  \node[draw=none, node font=\bfseries] (node4) at (2.41,2.7) {$K\times D_z $ 
        };
  \node[draw=none] at (12.56,2.52) {$K\times D_z$ 
          };
  \node[draw=none] at (15.71,3.18) {Output: 
          };
  \node[draw=none] (node5) at (15.87,2.7) {$P\times D$ 
            };
  \node[draw=none, text=green!50!black, node font=\large\bfseries] at (14.87,6.15) {Receiver (Cloud/Edge Server)};
  \draw[teal, fill=teal!70] (3.66,4.9) coordinate (A1) rectangle ++(0.6,0.4) coordinate (B1);
  \node[draw=none, node font=\footnotesize\bfseries] at ($(A1)!0.5!(B1)$) {$\mathbf{z}_{e, 1}$
  
            };
    
  \draw[teal, fill=teal!42] (3.66,4.25) coordinate (A2) rectangle ++(0.6,0.4) coordinate (B2);
  \node[draw=none, node font=\footnotesize\bfseries] at ($(A2)!0.5!(B2)$) {$\mathbf{z}_{e, 2}$
  
            };
      
  \draw[teal, fill=teal!22] (3.66,3.6) coordinate (A3) rectangle ++(0.6,0.4) coordinate (B3);
  \node[draw=none, node font=\footnotesize\bfseries] at ($(A3)!0.5!(B3)$) {$\mathbf{z}_{e, 3}$
  
            };
        
  \draw[teal, fill=teal!10] (3.66,2.95) coordinate (A4) rectangle ++(0.6,0.4) coordinate (B4);
  \node[draw=none, node font=\footnotesize\bfseries] at ($(A4)!0.5!(B4)$) {$\mathbf{z}_{e, 4}$
  
            };
          
  \draw[teal, fill=teal!5] (3.66,2.2) coordinate (A5) rectangle ++(0.6,0.4) coordinate (B5);
  \node[draw=none, node font=\footnotesize\bfseries] at ($(A5)!0.5!(B5)$) {$\mathbf{z}_{e,K}$
  
            };
  \node[draw=none, rotate=-90] at (3.96,2.8) {...};
    
  \draw[violet, fill=violet!70] (7.37,4.9) coordinate (C1) rectangle ++(0.6,0.4) coordinate (D1);
  \node[draw=none, node font=\footnotesize\bfseries] at ($(C1)!0.5!(D1)$) {$\mathbf{z}_{q, 1}$
  
            };
     
  \draw[violet, fill=violet!40] (7.37,4.25) coordinate (C2) rectangle ++(0.6,0.4) coordinate (D2);
  \node[draw=none, node font=\footnotesize\bfseries] at ($(C2)!0.5!(D2)$) {$\mathbf{z}_{q, 2}$
  
            };
        
  \draw[violet, fill=violet!20] (7.37,3.6) coordinate (C3) rectangle ++(0.6,0.4) coordinate (D3);
  \node[draw=none, node font=\footnotesize\bfseries] at ($(C3)!0.5!(D3)$) {$\mathbf{z}_{q, 3}$
  
            };
          

  \draw[violet, fill=violet!10] (7.37,2.95) coordinate (C4) rectangle ++(0.6,0.4) coordinate (D4);
  \node[draw=none, node font=\footnotesize\bfseries] at ($(C4)!0.5!(D4)$) {$\mathbf{z}_{q, 4}$
  
            };
            
  \draw[violet, fill=violet!5] (7.37,2.2) coordinate (C5) rectangle ++(0.6,0.4) coordinate (D5);
  \node[draw=none, node font=\footnotesize\bfseries] at ($(C5)!0.5!(D5)$) {$\mathbf{z}_{q, K}$
  
            };
  \node[draw=none, rotate=-90] at (7.67,2.8) {...};
  \node[draw=none] at (9.86,3.23) {...};
  \draw[violet,fill=violet!70] (8.64,3.48) rectangle (8.97,3.05);
  \draw[violet, fill=violet!40] (8.97,3.48) rectangle (9.3,3.05);
  \draw[violet, fill=violet!20] (9.29,3.48) rectangle (9.62,3.05);
  \draw[violet, fill=violet!5] (10.04,3.48) rectangle (10.37,3.05);
  \draw[fill=red!74, very thin] (9.86,3.65) -- (9.66,3.29) -- (9.77,3.18) -- (9.64,2.78) -- (9.84,3.18) -- (9.71,3.3) -- cycle;
  \node[draw=none, text=black, node font=\bfseries] at (9.42,2.04) {Truncation};
  \node[draw=none, node font=\bfseries, text=violet] at (12.95,5) {Codebook  
      };
  \node[draw=none, node font=\bfseries, text=violet] at (12.95,4) {Projection   
        };
  \node[draw=none, node font=\bfseries, text=violet] at (12.95,4.5) {Lookup /};
  \fill[blue!20, draw=blue, dashed] (14.22,0.37) rectangle (17.18,1.56);
  \fill[blue!20, draw=blue] (14.24,-0.5) rectangle (17.2,0.14);
  \node[draw=none, node font=\bfseries] at (15.73,1.1) {Backbone 
    };
  \node[draw=none, node font=\bfseries] (node6) at (15.69,-0.16) {Task
Head 
      };
  \node[draw=none, node font=\bfseries] at (15.79,0.68) {Late Encoder  
        };
  \node[draw=none, text=black, node font=\large\bfseries] at (9.55,6.15) {GO-OAE};
  \node[draw=none, text=blue, node font=\large\bfseries] at (2.59,6.06) {Transmitter (Device)};
  \draw[->, thick] (1.3,-0.02) -- (1.55,-0.02);
  \draw[->, thick] (3.05,1.43) -- (3.05,2.04);
  \draw[->, thick] (4.26,3.5) -- (4.64,3.5);
  \draw[->, thick] (7.91,3.5) -- (8.29,3.5);
  \draw[->, thick] (10.81,3.5) -- (11.66,3.5);
  \draw[->, thick] (14.32,3.5) -- (14.87,3.5);
  \draw[->, thick] (15.8,2.25) -- (15.81,1.55);
  \draw[->, thick] (17.2,0) -- (17.58,0);
  \draw[->, thick] (15.89,0.37) -- (15.89,0.14);
\end{tikzpicture}
}
\end{center}
\vspace{-1.5ex}
\caption{GO-OAE for progressive goal-oriented communication. 
The encoder maps $P$ input components of dimension $D$ (e.g., image patches or feature tokens) into an ordered latent sequence $\{\mathbf{z}_{e,i}\}_{i=1}^K$, which is vector-quantized using a jointly learned codebook, producing the quantized latent prefix $\{\mathbf{z}_{q,i}\}_{i=1}^K$. 
Transmitting the first $k\leq K$ codes (discrete indices of the mapped codebook embeddings) requires $k\cdot \log_2 N$ bits per sample. The receiver reconstructs the corresponding latent prefix and uses it directly for downstream task inference.}
\label{fig:go_oae_framework}
\vspace{-2ex}
\end{figure*}

Our main contributions are summarized as follows:
\begin{itemize}

\item \textbf{\ac{GO-OAE}: A new information-theoretic framework.}
We propose \ac{GO-OAE}, the first information-theoretic framework for progressive goal-oriented communication, which learns an ordered latent sequence directly optimized for downstream tasks. We show that its objective is equivalent to a weighted conditional \ac{IB} formulation, providing an information-theoretic foundation for progressive communication and semantic representation learning.
\item \textbf{Emergent ordered representation with efficient training.}
Unlike prior ordered learners that rely on autoregressive or architectural constraints \cite{xu2021anytime,liang2025semantic}, \ac{GO-OAE} derives ordering from a weighted conditional \ac{IB}-equivalent objective. We further develop an unbiased \ac{MPIS} estimator that reduces training complexity from $\Os(BK)$ to $\Os(Bm)$, with $m{\ll}K$, and reduces the estimator variance by a factor of $1/m$, where $B$ is the minibatch size, $K$ is the full prefix length, and $m$ is the prefix sampling budget.
\item \textbf{Goal-Oriented codec design principles.}
Our empirical analysis shows two key design principles for goal-oriented codecs: reconstruction quality does not necessarily translate to task performance, and conventional codebook objectives inadequately preserve semantic structure, motivating task-aware regularization.
\item \textbf{Progressive communication at extremely low rates.}
Experiments on MiniImageNet \cite{ravi2017optimization,russakovsky2015imagenet} show that \ac{GO-OAE} enables robust inference under arbitrary transmission truncation, achieving over $96\%$ of lossless classification accuracy using only $324$ bits per sample ($683{\times}$ compression). Across all transmission budgets, it consistently outperforms reconstruction-based codecs and ordered autoencoders trained with uniform prefix sampling \cite{xu2021anytime}.
\end{itemize}
The remainder of this paper is organized as follows. Section~\ref{sec:related_works} reviews related works. Section~\ref{sec:preliminaries} introduces the background on \acp{OAE} and \ac{VQ-AE}. Section~\ref{sec:go-oae} presents the proposed \ac{GO-OAE} framework, including its information-theoretic formulation. The proposed \ac{GO-OAE} algorithm is presented in Section~\ref{sec:proposed_algorithm}, including \ac{MPIS} estimator and loss design. Section~\ref{sec:numerical_results} provides numerical results. Finally, Section~\ref{sec:conclusions} concludes the paper.
\section{Related works}\label{sec:related_works}
\textbf{Goal-Oriented Semantic Communication.} 
Semantic communication shifts the objective from reliable bit transmission to preserving task-relevant information \cite{qin2021semantic}. Existing studies have investigated its information-theoretic foundations, AI-enabled architectures, end-to-end learning methods, and potential applications to future networks \cite{gunduzsemantic2023, getusemantic2024, luo_2022, liao2024adasem}. Recent works have increasingly incorporated generative models into image and video semantic communication. Diffusion-based methods use semantic extraction and generative reconstruction to improve perceptual robustness under channel impairments \cite{li2026goal}. Self-supervised approaches learn task-agnostic latent representations that are reconstructed before classification \cite{wanasekara2025sc}, while multi-task frameworks jointly support image reconstruction and classification from shared semantic features \cite{lyu2024semantic}. 

Despite these advances, most existing methods remain reconstruction-oriented: task inference is performed on reconstructed signals, and progressive transmission primarily improves perceptual fidelity. Moreover, hierarchical representations do not necessarily induce an ordered latent space in which every prefix is directly optimized for task utility.

\textbf{Ordered Autoencoding and Successive Refinement.}
Successive refinement is a fundamental concept in information theory, allowing progressive transmission of information with increasing levels of detail \cite{equitz1991successive}. Learned \acp{OAE} similarly order latent variables so that longer prefixes improve reconstruction \cite{xu2021anytime}, while autoregressive methods predict the missing suffix from the received prefix \cite{li2024autoregressive}. However, these approaches optimize reconstruction or generative completion rather than task utility. Let $\vz_{\leq k}$ denote the transmitted prefix with length $k$, and $\widehat{\vz}_{>k} \sim p\left(\widehat{\vz}_{>k}{\mid}\vz_{\leq k}\right)$ an autoregressively generated suffix. The suffix is generated solely from the received prefix, 
i.e., $
I\left(\vy;\vz_{\leq k},\widehat{\vz}_{>k}\right){=}I\left(\vy;\vz_{\leq k}\right).
$
Hence, autoregressive approaches learn a conditional prior over the untransmitted latent suffix, enabling plausible latent completion. However, it cannot add task information beyond the received prefix and also incurs sequential decoding overhead.

Another relevant work is \ac{MRL} \cite{kusupati2022matryoshka}, which learns nested prefixes of a single final embedding after full backbone processing, placing most computation at the transmitter. It targets scalable representation storage and inference rather than lightweight-device computation and communication reduction. In contrast, \ac{GO-OAE} orders and discretizes intermediate latent sequences before full semantic reasoning, enabling a lightweight transmitter and progressive receiver-side inference.
\section{Definitions and Preliminaries}\label{sec:preliminaries}

Let $\vx\in\Xs$ denote an input sample which follows the distribution $p_{\vx}(\vx)$. An encoder $e_{\theta}:\Xs\to\Zs^K$ maps $\vx$ to a latent sequence $\vz=(\vz_1, \ldots, \vz_K)$,  where each latent component $\vz_i\in\Zs$, for $i\in\Ks:=\{1, \ldots,K\}$. The latent space $\Zs$ is generic. For example, in transformer-based architectures each $\vz_i$ corresponds to a token embedding. A reconstruction decoder $d_{\phi}:\Zs^K\to\Xs$ maps the latent sequence to the reconstructed input $\widehat{\vx}$  (the goal-oriented decoder is presented in Section \ref{sec:go-oae}). For ease of reference, the main notations used throughout this paper are summarized in Table~\ref{tab:notations}. 
\begin{definition}[Prefix Length and Latent Prefix]
    Let $k\in\Ks$ denote the \emph{prefix length}, i.e., the number of retained latent components after truncating the latent sequence. For a given prefix length $k$, the \emph{latent prefix} is defined as $\vz_{\leq k}:=(\vz_1, \dots, \vz_k)\in\Zs^k$, which consists of the first $k$ components. For implementation purposes, $\vz_{\leq k}$ may equivalently denote its zero-padded form $(\vz_1, \dots, \vz_k, 0, \dots, 0)\in\Zs^K$, where the remaining components are set to zero. We further define $\vz_{< k} = (\vz_1, \dots, \vz_{k-1})$ which excludes the $k$-th component.
\end{definition}

\begin{table}[t] 
    \centering 
    \caption{Summary of Key Notations} 
    \vspace{-1.5ex}
    \label{tab:notations} 
    \begin{tabular}{cl} 
        \toprule 
        \textbf{Symbol} & \textbf{Description} \\ 
        \midrule 
        $\vz$ & Latent sequence $(\vz_1, \ldots, \vz_K)\in\Zs^K$ \\ 
        $i$ & Latent component index, also as sampled or corresponding \\
        & prefix length when defining the ordered objective \\
        $k$ & Selected prefix length \\
        $\vz_{\leq k}$ & Latent prefix of length $k$, $(\vz_1, \ldots, \vz_k)\in\Zs^k$\\  
        $\vz_i$ & The $i$-th latent component, $\vz_i\in\Zs$ \\ 
        $\vz_{e,i}$ & Encoder latent component, $\vz_{e,i}\in\Zs_e$ \\ 
        $\vz_{q,i}$ & Quantized latent component, $\vz_{q,i}\in\Zs_q$ \\
        $\vz_{q, \leq k}$ & Quantized latent prefix of length $k$, $\vz_{q, \leq k}\in\Zs_q^k$ \\
        $\ve{e}_j$ & Codebook embedding, $\ve{e}_j\in\Es$ \\ 
        $j_i$ & Codebook index for the $i$-th component (referred to \\ &interchangeably as ``code index'' throughout the paper)\\ 
        $N$ & Codebook size, $|\Es|=N$ \\ 
        $m$ & Prefix sampling budget, $m{\ll}K$ \\ 
        \bottomrule 
    \end{tabular} 
    \vspace{-4ex}
\end{table}

\subsection{Ordered Autoencoder for Reconstruction}\label{ssec:OAE_recon}
To encourage progressively informative latent components, we define the \ac{OAE} as follows. For notational simplicity, instead of $k$, we use the latent index $i$ to also denote the corresponding prefix length in the ordered objective. 
\begin{definition}[Reconstruction-Based \ac{OAE}]\label{def:recon-OAE}
    An encoder-decoder pair $(e_{\theta}, d_{\phi})$ is an ordered autoencoder if its latent sequence $\vz=(\vz_1,\ldots,\vz_K)$ has non-increasing incremental information gains about the input. In particular, define
    \vspace{-.5ex}
\begin{equation}
\label{eqn}
\Delta_i
:=
I(\vz_i;\vx\mid \vz_{<i}),
\qquad i\in\Ks,
\vspace{-.5ex}
\end{equation}
where $\vz_{<i}:=(\vz_1,\ldots,\vz_{i-1})$ and $I(\cdot;\cdot\mid\cdot)$ denotes conditional mutual information. The ordered latent sequence satisfies
\vspace{-.5ex}
\begin{equation}
\Delta_1\geq\Delta_2\geq\cdots\geq\Delta_K\geq0.
\label{eqn:MI_OAE}
\vspace{-.5ex}
\end{equation}
\end{definition}

Xu et al. \cite{xu2021anytime} show that the ordered autoencoder condition can be achieved by minimizing the objective below.
\begin{theorem}[Ordered Reconstruction Objective]
\label{theorem:ordered_reconstruction}
\cite[Theorem 1]{xu2021anytime}
Let the \ac{OAE} reconstruction objective be defined as:
\begin{equation}
    \label{eqn:objective_OAE_1}
    \begin{split}
    \Ls_{\text{OAE}}(e_{\theta}, d_{\phi}) := &\; \EE_{\vx\sim p_{\vx}(\vx), i\sim \Us\{1, K\}} \\
    & \EE_{\vz\sim q_{e_\theta}(\vz_{\leq i}|\vx)} \left[-\log p_{d_\phi}(\vx|\vz_{\leq i})\right],
    \end{split}
\end{equation}
where $\Us\{1, K\}$ denotes the uniform distribution over the set $\{1, \ldots, K\}$, $q_{e_\theta}(\vz_{\leq i}|\vx)$ and $p_{d_\phi}(\vx|\vz_{\leq i})$ denote the distributions of the encoder and the decoder with prefix length $i\in\Ks$, respectively. For a fixed-variance Gaussian decoder, and if the encoder satisfies conditional-independence, namely $\vz_{i-1}{\perp}\vz_i
{\mid}\vx,\vz_{\leq i-2}$, for
$i\geq 3$, then, \eqref{eqn:objective_OAE_1} is equivalent, up to constants, to the prefix-averaged mean squared reconstruction loss. An optimal encoder-decoder pair $(e_{\theta^\ast}, d_{\phi^\ast})$ that minimizes \eqref{eqn:objective_OAE_1} satisfies the ordered autoencoder condition in \eqref{eqn:MI_OAE}. The finite deterministic loss over $N_s$ samples is
\begin{equation}
    \label{eqn:objective_OAE_2}
    \frac{1}{N_s}\sum_{n=1}^{N_s}\frac{1}{K}\sum_{i=1}^K \|\vx_n - d_{\phi}(\vz_{\leq i})\|^2.
\end{equation}
\end{theorem}

Because \eqref{eqn:objective_OAE_1} is an expectation over prefix length following uniform distribution, instead of summing over all $i$, it can be estimated unbiasedly by sampling a single prefix length per minibatch and optimizing the reconstruction loss using $\vz_{\leq i}$.

\subsection{Vector Quantized Autoencoder}\label{ssec:VQ-VAE}
We adopt \ac{VQ-AE} variant of \ac{VQ-VAE} \cite{van2017neural} as the underlying discrete latent modeling framework for ordered autoencoding. \ac{VQ-AE} learns discrete latent space by mapping the continuous encoder outputs to a finite set of learnable codebook embeddings through vector quantization. The resulting codebook indices are mapped back to their corresponding embeddings, from which the decoder reconstructs the input. 

We choose \ac{VQ-AE} because it naturally supports ordered and progressively refinable latent sequence. First, discretized latent spaces provide more stable and semantically structured representations than continuous variables under progressive truncation, where discrete codebook indices can be transmitted or decoded in an ordered manner. Second, transmitting discrete indices is far more communication-efficient and implementation-friendly than continuous latent vectors. For example, Llama~3 8B model \cite{grattafiori2024llama} represents each token with a $4096$-dimensional FP16 vector, requiring $65,536$ bits per token, whereas indexing a codebook of size $128,256$ requires only $17$ bits. In vision tasks, even smaller codebooks are typically sufficient, e.g., 512 entries, with each latent component index (or "code index") represented by only $9$ bits. 

Let the $i$-th encoded and quantized latent components be denoted by $\vz_{e,i}{\in}\Zs_e$ and $\vz_{q,i}{\in}\Zs_q$, respectively. The finite codebook is denoted by $\Es{:=}\{\ve{e}_j\}_{j=1}^N$, with $N$ learnable codebook embeddings $\ve{e}_j, j=1, \ldots, N$. The quantization is performed via nearest neighbor lookup, i.e., for each $i\in\Ks$, we have
\begin{equation}
    \vz_{q,i} = \ve{e}_{j_i}, \quad \text{where } j_i = \argmin_{j\in\{1, ..., N\}} \|\vz_{e,i} - \ve{e}_{j}\|_2^2.
    \label{eqn:quantization}
\end{equation}
The quantized latent indices are then provided to the decoder.

\begin{remark}[Latent and Codebook Embedding Dimensions] For notational simplicity, in this section, we assume that $\vz_{e,i}$, $\vz_{q,i}$, and  $\ve{e}_j$ have the same dimension $D_z$, as in standard \ac{VQ-AE}, so that nearest-neighbor quantization is directly performed in the latent space. However, more generally, a learnable projection may be inserted before quantization to map encoder outputs to a latent space of different dimensionality, as done in our experiments in Section \ref{sec:numerical_results}.
    \label{rem:latent_dim}
\end{remark}

To train \ac{VQ-AE}, we minimize the following loss:
\begin{equation}   
    \Ls_{\text{VQ}} = \Ls_{\text{rec}} + \Ls_{\text{codebook}} + \beta \Ls_{\text{commit}},
    \label{eqn:loss_vqvae}
\end{equation}
where $\Ls_{\text{rec}}=-\log p_{d_\phi}(\vx|\vz_q)$ is the reconstruction loss, $\Ls_{\text{codebook}} = \|\text{sg}[\vz_e] - \vz_q\|_2^2$ is the codebook loss that updates the codebook embeddings, and $\Ls_{\text{commit}} = \|\vz_e - \text{sg}[\vz_q]\|_2^2$ is the commitment loss that encourages the encoder outputs to commit to the quantized latents, where $\text{sg}[\cdot]$ denotes the stop-gradient operator, and $\beta$ is a hyperparameter that balances the losses. Here, we omit the indices $i$ for brevity. 

In vanilla \ac{VQ-AE}, sparse code assignments often lead to poor code utilization or even codebook collapse, where only a few entries are frequently used. Such behavior limits the expressive capacity of the latent space and may cause semantically dissimilar inputs to share similar code indices. To mitigate this, we adopt the \ac{EMA}-based \ac{VQ-VAE} \cite{lancucki2020robust} update scheme for our autoencoder, which updates codebook embedding using exponential moving average of assigned encoder outputs. This performs an online centroid update and improves codebook stability.

\section{Goal-Oriented Ordered Autoencoding}\label{sec:go-oae}
While the ordered autoencoders introduced in Section \ref{sec:preliminaries} learn progressively refinable latents to reconstruct input data, they may not explicitly preserve the semantics for the downstream tasks. Surprisingly, even with high reconstruction fidelity, the latent representation may not be informative for the downstream task, as shown in Section \ref{sssec:reconstruction_fidelity}. Motivated by this limitation, we extend \ac{OAE} toward a progressive goal-oriented communications framework \ac{GO-OAE}. It consists of an encoder $e_{\theta}$ and a task-specific decoder $d_{\phi}:\Zs^K \to \Ys$ (or $\Zs_q^K \to \Ys$ if vector quantization is applied) where the objective is to optimize the task utility.
\begin{definition}[\ac{GO-OAE}]
    A goal-oriented ordered autoencoder is defined as an encoder-decoder pair $(e_{\theta}, d_{\phi})$, optimized for minimizing the task loss across all prefix lengths:
    \begin{equation}
        \min_{\theta, \phi} \EE_{p(\vx, \vy)}\left[\sum_{i=1}^K \Ls(d_{\phi}(\vz_{\leq i}), \vy)\right],
        \label{eqn:GO_OAE_objective}
    \end{equation}
    where $\mathcal{L}$ is a loss function that measures the discrepancy between the decoded output and the true task label. 
\end{definition}

Throughout this paper, we instantiate the generic loss $\Ls$ in \eqref{eqn:GO_OAE_objective} as the negative log-likelihood (cross-entropy) of the task label, which yields the \ac{GO-OAE} distortion function $\Ls_D$ defined later in \eqref{eqn:goal_oriented_distortion}. In the following subsections, we will establish a theoretical equivalence between the \ac{GO-OAE} distortion objective and a weighted conditional \ac{IB} functional under ordered latent truncation, and further show that the learned latent representation satisfies a monotonically decreasing conditional mutual information ordering.

\subsection{Connection to Information Bottleneck}\label{ssec:GO_OAE_IB}
First, we show that the \ac{GO-OAE} distortion function can be connected to a weighted conditional \ac{IB} functional under ordered latent truncation.

\begin{theorem}[\ac{GO-OAE} Distortion Objective]
    \label{theorem:GO_OAE_objective}
    Assume the Markov structure $\vy{\to}\vx{\to}\vz$ holds, meaning that latent sequence $\vz$ is generated from the input, and task label influences $\vz$ only through $\vx$, allowing $p(\vz|\vx, \vy)=p(\vz|\vx)$, and assume Bayes-optimal task decoder such that $\E{-\log p(\vy|\vz_{\leq i})} = H(\vy|\vz_{\leq i})$, where $H(\cdot{\mid}\cdot)$ denotes the conditional entropy. We define the \ac{GO-OAE} distortion function as
    \begin{equation}
    \label{eqn:goal_oriented_distortion}
    \Ls_D := \sum_{i=1}^K\EE_{p(\vx, \vy)}\left[-\log p(\vy|\vz_{\leq i})\right].
\end{equation}   
    Then minimizing $\Ls_D(\theta, \phi)$ is equivalent, up to the constant  $K H(\vy)$, to maximizing
    \begin{equation}
        \sum_{i=1}^K w_i I(\vy; \vz_i|\vz_{< i}),
        \mbox{ with } w_i = K - i+1. 
        \label{eqn:GO_OAE_IB} 
    \end{equation}  
\end{theorem}  
\begin{proof}
The proof is given in Appendix \ref{ssec:proof_GO_OAE_objective}.
\end{proof}

The Markov structure $\vy{\to}\vx{\to}\vz$ guarantees that $I(\vy; \vz)\leq I(\vy;\vx)$, i.e., $\vz$ cannot contain task-relevant information beyond what is contained in $\vx$. The goal is not to create new task information, but to discard irrelevant details through the bottleneck term $I(\vz; \vx)$. If we further regularize the latent sequence by its information complexity $I(\vz; \vx)$, the resulting objective becomes a weighted conditional \ac{IB} functional under ordered latent truncation to maximize:
\vspace{-.5ex}
\begin{equation}    
\sum_{i=1}^K w_i I(\vy; \vz_i|\vz_{< i}) - \lambda I(\vz; \vx),
\label{eqn:GO_OAE_IB_regularized}
\vspace{-.5ex}
\end{equation}
where $\lambda$ is a coefficient that controls the tradeoff between task performance and representation compression.

\subsection{Goal-Oriented Latent Ordering}\label{ssec:GO_OAE_ordering}
We further provide the following corollary to formally characterize the ordering property of the \ac{GO-OAE} objective.
\begin{corollary}[Goal-Oriented Latent Ordering]
    \label{corollary:GO_OAE_ordering}
    Assume the Bayes-optimal encoder-decoder pair $(e_{\theta}, d_{\phi})$ is an optimal solution that minimizes the \ac{GO-OAE} distortion objective \eqref{eqn:goal_oriented_distortion}, and each latent component has an intrinsic task-relevant contribution that is invariant under permutation of latent positions. Then, the learned latent sequence satisfies:
    \begin{equation}
        I(\vy; \vz_1) \geq I(\vy; \vz_2|\vz_1) \geq \ldots \geq I(\vy; \vz_K|\vz_{<K}).
    \end{equation}
\end{corollary}
\begin{proof}
The proof is given in Appendix \ref{ssec:proof_GO_OAE_ordering}.
\end{proof}

\begin{remark}[Information Bottleneck Behavior and Latent Ordering]
\label{rem:ordered_IB}
Even without explicitly including $I(\vz;\vx)$ in \eqref{eqn:goal_oriented_distortion}, the \ac{GO-OAE} objective exhibits an implicit information bottleneck effect through the \ac{VQ-AE} bottleneck, finite codebook, latent dimensionality, and stochastic truncation. Beyond this implicit compression, the latent ordering is a natural consequence of the objective: earlier components appear in more prefixes and thus receive larger effective weights $w_i = K-i+1$. Moreover, as formally established in Corollary~\ref{corollary:GO_OAE_ordering}, this yields a monotonically decreasing conditional mutual information ordering, with earlier dimensions capturing the most critical task-relevant semantics and later ones providing finer refinements.
\end{remark}
      
\section{Proposed \ac{GO-OAE} Algorithm}\label{sec:proposed_algorithm}

In the objective of reconstruction-based \ac{OAE} \eqref{eqn:objective_OAE_1}, the prefix length yields a uniformly random distribution $\Us\{1,K\}$, leading to an unbiased estimation of the prefix length-averaged objective. This enables efficient training by randomly sampling a single prefix length for each minibatch. In contrast, the \ac{GO-OAE} objective \eqref{eqn:goal_oriented_distortion} assigns inherently \textbf{non-uniform importance} to different prefix lengths, since \eqref{eqn:GO_OAE_IB} introduces prefix-dependent weights $w_i = K-i+1$. Therefore, uniformly sampling prefix lengths no longer provides an efficient optimization strategy, motivating a new algorithm that accounts for non-uniform importance of different prefix lengths.

Objective \eqref{eqn:goal_oriented_distortion} involves a summation over all prefix lengths. Thus, for each minibatch with size $B$, it becomes 
\begin{align}
    \Ls_D & = \frac{1}{B}\sum_{b=1}^B \sum_{i=1}^K l_i^{(b)}, \label{eqn:summed_loss}\\
    \mbox{ with } 
    l_i^{(b)} & = -\log p\left(\vy^{(b)}|\vz_{\leq i}^{(b)}\right),\label{eqn:loss_per_level}
\end{align}  
where $\vy^{(b)}$ and $\vz_{\leq i}^{(b)}$ are the label and latent prefix (with length $i$) of the $b$-th sample in the minibatch, respectively. This results in $\Os(BK)$ decoder evaluations per iteration. Since $l_i^{(b)}$ must be computed for all $i$ and $b$, it can be computationally expensive for large $K$. Conversely, naively sampling a single prefix length $i$ per minibatch, as in reconstruction-based \ac{OAE} \cite{xu2021anytime}, reduces the complexity to $\Os(B)$ but may introduce unbalanced supervision across prefix lengths. This motivates a stochastic approximation of \eqref{eqn:goal_oriented_distortion} that is (i) unbiased, (ii) computationally efficient, and (iii) with low variance for stable training.

\subsection{Unbiased Multi-Prefix Importance Sampling Estimator}
We now seek a sampling strategy that approximates \eqref{eqn:summed_loss} with fewer decoder evaluations while preserving unbiasedness and controlling the variance of the estimator. The following corollary formalizes the \ac{MPIS} estimator and its variance.

\begin{corollary}[\ac{MPIS} Estimator]
    \label{corollary:unbiased_estimator}
Let the \ac{GO-OAE} objective be
\begin{equation}
\Ls_D  = \sum_{i=1}^K L_i \mbox{ where } L_i = \EE_{p(\vx, \vy)}\left[-\log p\left(\vy|\vz_{\leq i}\right)\right].
\end{equation}
Let $q(i)$ be a sampling distribution over the prefix length $i$ such that $q(i)>0,~\sum_{i=1}^K q(i) = 1$. 
At each minibatch, we sample $m$ prefix lengths $\{i_1, i_2, \ldots, i_m\}$ independently from $q(i)$, and compute the stochastic estimator of $\Ls_D$: 
\begin{equation}
    \widehat{\Ls}_D = \frac{1}{m} \sum_{n=1}^m \frac{L_{i_n}}{q(i_n)}.
    \label{eqn:stochastic_estimator_m}
\end{equation}
Then, $\widehat{\Ls}_D$ is an unbiased estimator of $\Ls_D$, i.e., 
\begin{equation}
\label{eqn:unbiased_estimation}
    \EE_{i_1, \ldots, i_m \sim q}\left[\widehat{\Ls}_D\right] = \Ls_D.
\end{equation}
The estimator variance yields:
\begin{equation}
\label{eqn:stochastic_estimator_variance}
    \mathrm{Var}\left(\widehat{\Ls}_D\right) = \frac{1}{m} \left(\sum_{i=1}^K \frac{L_i^2}{q(i)} -  \Ls_D^2\right).
\end{equation}
Moreover, assuming differentiation and expectation can be interchanged, stochastic gradient descent using $\widehat{\Ls}_D$ optimizes the same objective \eqref{eqn:goal_oriented_distortion}, but with a reduced complexity $\Os(Bm)$.
\end{corollary}

\begin{proof}
The proof is given in Appendix \ref{ssec:proof_unbiased_estimator}.
\end{proof}

Corollary \ref{corollary:unbiased_estimator} allows us to develop an unbiased and efficient training strategy for \ac{GO-OAE} by  significantly reducing the computational cost from $\Os(BK)$ to $\Os(Bm)$, where $m{\ll}K$. We also discuss why single-prefix sampling as in \cite{xu2021anytime} may introduce bias and unstable training in following Remark \ref{remark:single_truncation_level_sampling}.

\begin{remark}[Single-Prefix Sampling]
    \label{remark:single_truncation_level_sampling}
    Sampling a single prefix length $i{\sim}q(i)$ and directly use $L_i$ is generally biased, because $\E{L_i}=\sum_{i=1}^K q(i)L_i$ is not equal to $\Ls_D = \sum_{i=1}^K L_i$. Applying importance weight to single prefix sampling $\widehat{\Ls}_D = L_i/q(i)$, for $i{\sim} q(i)$ is unbiased, but has a high variance $\Var\left[\widehat{\Ls}_D\right] = \sum_{i=1}^K L_i^2/q(i) - \Ls_D^2$, which can be large when a high-loss prefix has low sampling probability. Our proposed \ac{MPIS} in Corollary~\ref{corollary:unbiased_estimator} reduces this variance by a factor of $1/m$ as implied in \eqref{eqn:stochastic_estimator_variance}.
\end{remark}

The remaining challenge is to further reduce the variance for stable training. By leveraging \eqref{eqn:GO_OAE_IB}, we can design the sampling distribution as 
\vspace{-1ex}
\begin{equation}
    q(i)\propto \frac{w_i^\alpha}{\sum_{i=1}^K w_i^\alpha}, ~\mbox{for } i\in\Ks,
    \label{eqn:importance_sampling}
\end{equation}
where $\alpha\geq 0$ controls the degree of importance sampling, the higher the $\alpha$, the more focus on the early latent components with higher importance weights. This sampling strategy reduces the variance by allocating more samples to important components, while maintaining unbiasedness. In practice, $\alpha$ and prefix sampling budget $m$ are tuned for stable and efficient training.

\subsection{Overall Training Loss and Algorithm}
To learn an ordered latent space within the \ac{VQ-AE} framework (introduced in Section \ref{ssec:VQ-VAE}), we jointly optimize the encoder-decoder and codebook parameters by minimizing the overall \ac{GO-OAE} loss, consisting of unbiased stochastic goal-oriented loss $\widehat{\Ls}_D$ \eqref{eqn:stochastic_estimator_m}, the \ac{VQ} loss $\Ls_{VQ}$ \eqref{eqn:loss_vqvae}, and optional regularizer $\Rs$ (e.g.,  sparsity, disentanglement, codebook utilization, contrastive learning, or weight decay). The overall training loss for \ac{GO-OAE} is given by
\vspace{-.5ex}
\begin{equation}    
    \Ls_{\text{GO-OAE}} = \widehat{\Ls}_D + \beta \Ls_{\text{VQ}} + \lambda_{\text{reg}} \Rs,
  \vspace{-.5ex}
    \label{eqn:overall_loss}
\end{equation}
where the hyperparameters $\beta$ and $\lambda_{\text{reg}}$ control the balance. 

Note that the \textbf{unbiased estimator loss} $\widehat{\Ls}_D$ in \eqref{eqn:stochastic_estimator_m} may cause high-variance gradients when dividing by $q(i)$ for rare indices. For practical concern, to stabilize the training, we can clip the importance weights, by defining:
\vspace{-.5ex}
\begin{equation}
    \widehat{\Ls}_D = \frac{1}{m} \sum_{n=1}^m L_{i_n} \cdot \min\left(\frac{1}{q(i_n)}, C\right),
    \label{eqn:stochastic_estimator_m_clipped}  
    \vspace{-.5ex}
\end{equation}
where constant $C$ controls the maximum importance weight.

As for \textbf{\ac{VQ} loss} $\Ls_{\text{VQ}}$, as explained in Section \ref{ssec:VQ-VAE}, we adopt the \ac{EMA}-based update scheme~\cite{lancucki2020robust} to avoid codebook collapse. Codebook loss is not optimized by gradients, but rather updated by an exponential moving average of the encoder outputs assigned to each code. Thus, $\Ls_{\text{VQ}}=\|\vz_e - \text{sg}[\vz_q]\|_2^2$ denotes the commitment loss, encouraging encoder outputs to match the codebook embeddings. 

Below, we discuss two specific \textbf{regularization losses} in $\Rs$ to enhance semantic structure and codebook utilization.

{\bf Semantic Preservation Regularization:} To further preserve semantic structure in the ordered latent space, we optionally include a task-specific regularizer. For example, for classification, this can be a \emph{supervised contrastive loss} that pulls samples with the same class label closer in the latent space. Let $s_{ij}:=\left(\vz_{e,i}^T\vz_{e,j}\right)/\tau$ denote the similarity between pre-quant latent representations of samples $i$ and $j$ (with a slight abuse of notation), where $\tau$ is a temperature parameter. Let $\tau_0$ be the base temperature. For anchor sample $i$, let $\Sss(i)=\{j:\vy_j=\vy_i, j\neq i\}$ denote the set of positive samples sharing the same label as $i$, and $\Vs:=\{i:|\Sss(i)|>0\}$ denote the set of anchor samples with at least one positive sample. Then, the supervised contrastive loss is formulated as \cite{khosla2020supervised}:
\begin{equation}
    \Rs_{\text{SC}}=-\frac{\tau}{\tau_0}\cdot\frac{1}{|\Vs|}\sum_{i\in\Vs}\frac{1}{|\Sss(i)|}\sum_{j\in\Sss(i)}\log\frac{\exp(s_{ij})}{\sum_{k\neq i}\exp(s_{ik})}.
    \label{eqn:supervised_contrastive_loss}
\end{equation}

{\bf Entropy Regularization:} We include an entropy regularizer that increases the entropy of the empirical codebook-index distribution, encouraging balanced use of the codebook embeddings. Let $\vh{\in}\R^N$ be the empirical \ac{PMF} estimate across a minibatch, 
where $h_j{:=}n_j/T$ denotes the empirical frequency of code index $j$ being selected, $n_j$ is the number of times $j$ is selected in the current minibatch (with slight abuse of notation), and $T{:=}\sum_{j=1}^N n_j$ is the total number of sampled code indices.
Then, the regularizer to minimize is defined as:
\vspace{-1ex}
\begin{equation}
    \Rs_{\text{ent}} = \sum_{j=1}^N h_j \log (N h_j+ \epsilon),
    \label{eqn:entropy_regularization} 
    \vspace{-1ex}
\end{equation}
where $\epsilon$ is a small constant to prevent numerical instability. Note that $\log N h_j{=}\log h_j/u_j$, where $u_j{=}1/N, \forall j$, and $\vu{:=}(u_1, \ldots, u_N)$ is the uniform distribution over codebook indices. Thus, this term equals $D_{\text{KL}}(\vh\|\vu)$.

We optionally include the supervised contrastive and entropy regularizers $\{\Rs_{\text{SC}}, \Rs_{\text{ent}}\}$. The overall training algorithm is summarized in Algorithm \ref{alg:go_ovqae}. 
In \ac{EMA} update,  $\rho_j$ is a scalar running \ac{EMA} cluster size (cross-minibatch average of $n_j$) and $\vnu_j{\in}\R^{D_z}$ is an \ac{EMA} of the corresponding sum of assigned encoder outputs.
To stabilize training, we use \textbf{three-stages}: full-code pretraining, ordered training with \ac{MPIS}, and a final calibration stage with occasional full-code supervision to avoid under-training later components.

\begin{algorithm}[t]
\caption{\ac{GO-OAE} Training with \ac{MPIS} Estimator}
\label{alg:go_ovqae}
\begin{algorithmic}[1]
\Require Dataset $\Ds{:=}\{\Xs, \Ys\}$, encoder-decoder pair $(e_{\theta}, d_{\phi})$, prefix length $K$, codebook size $N$, randomly initialized codebook $\Cs^{(0)}{=}\left\{\vc_j^{(0)}\right\}_{j=1}^{N}$, prefix sampling budget $m{\ll}K$, batch size $B$, commitment weight $\beta$, regularization weight $\lambda_{\text{reg}}$, \ac{EMA} decay $\gamma$, sampling sharpness $\alpha$
\State Initialize EMA cluster size $\rho_j{\leftarrow}0$ and \ac{EMA} sums of assigned latent components $\vnu_j{\leftarrow}0$ for all code index $j$
\State Precompute sampling distribution $\{q(i)\}_{i=1}^K$ based on \eqref{eqn:importance_sampling}
\For{Each epoch}
    \For{Each minibatch, samples $\{(\vx^{(b)}, \vy^{(b)})\}_{b=1}^B$}
        \State {\bf Encode}: $\vz_e^{(b)} = e_{\theta}\left(\vx^{(b)}\right)$, $\forall b$
        \State {\bf Quantize}: Compute $\vz_q^{(b)}, \forall b$ with \eqref{eqn:quantization}
        \State $\begin{aligned}[t] & \text{{\bf Sample:} $m$ prefix lengths: } \{i_1, i_2, \ldots, i_m\} \\
        &  \text{from } q(i) ~\text{independently}\end{aligned}$
        \For{each sampled prefix length $i_n$}
            \State {\bf Truncate}: $\vz^{(b)}_{q, \leq i_n}$, $\forall b$
            \State {\bf Decode}: $\widehat{\vy}_{i_n}^{(b)} = d_{\phi}\left(\vz^{(b)}_{q, \leq i_n}\right)$, $\forall b$
            \State  Compute GO loss: $l_{i_n}^{(b)}$, $\forall b$ with \eqref{eqn:loss_per_level}
        \EndFor
        \State {\bf Compute per-minibatch losses:}
        \State Compute stochastic estimator: $\widehat{\Ls}_D$ with \eqref{eqn:stochastic_estimator_m_clipped}
        \State Compute commitment loss $\Ls_{\text{VQ}}$, e.g., \eqref{eqn:loss_vqvae}
        \State $\begin{aligned}[t] & \text{Compute regularization term } \Rs, \text{e.g., } \Rs_{\text{SC}}~\eqref{eqn:supervised_contrastive_loss}\\ & \text{ and } \Rs_{\text{ent}}~\eqref{eqn:entropy_regularization}
        \end{aligned}$
        \State $\begin{aligned}[t] & \text{{\bf Update encoder-decoder parameters}} \{\theta, \phi\} \text{ by}\\
            & \text{minimizing } \eqref{eqn:overall_loss} \text{ with stochastic gradient descent}
        \end{aligned}$
        \State $\begin{aligned}[t] & \text{{\bf Update codebook embeddings} using EMA \cite{lancucki2020robust}:}\\
            & \text{for each code index $j$, update $\rho_j$ and $\vnu_j$ based} \\
            & \text{on assigned outputs, update } \vc_j{\leftarrow}\vnu_j / \rho_j, \forall j
        \end{aligned}$
    \EndFor
\EndFor            
\end{algorithmic}
\end{algorithm}

\section{Experimental Results}\label{sec:numerical_results}
We evaluate \ac{GO-OAE} on image classification and target a tight budget of $K{=}36$ maximum prefix length (maximum number of transmitted code indices) and codebook size $N{=}512$, although the framework applies to other downstream tasks. Full \ac{GO-OAE} transmission requires only $36\log_2 512{=}324$ bits/image, i.e., $683{\times}$ fewer than raw $96{\times}96$ RGB storage ($96^2{\times}3{\times}8{=}221{,}184$ bits/image). Classical compressors (e.g., JPEG, typically ${\sim}11$-$68{\times}$, ${\sim}8$-$20$k bits/image) use far more bits and optimize pixel fidelity rather than task sufficiency.

We compare \textbf{eight scenarios} with the same backbone and data augmentation, which differ in quantization, goal-oriented training, supervised contrastive (SupCon) semantic regularization, and prefix sampling. Short names are used below.
\begin{itemize}
    \item \textbf{UB:} Upper bound, task model on raw images.
    \item \textbf{RE-OAE:} Reconstruct-then-classify, frozen \textbf{UB} classifier on reconstructed images. Same as training an \ac{OAE} in Definition \ref{def:recon-OAE} with objective \eqref{eqn:objective_OAE_1} in \cite{xu2021anytime}, but without task-awareness.
    \item \textbf{RE-OAE-TA:} Enhanced \textbf{RE-OAE} with task-awareness by including task utility in loss.
    \item \textbf{GO-VQAE:} Goal-oriented unordered \ac{VQ-AE} + task head (all $K$ quantized latent components), without ordering fine-tuning.
    \item \textbf{GO-VQAE-SC:} Enhanced \textbf{GO-VQAE} with SupCon loss \eqref{eqn:supervised_contrastive_loss} on pre-quantization features.
    \item \textbf{GO-OAE-U:} Ordered training with uniform sampling $k\sim\mathcal{U}\{1,K\}$ and full task-codec loss \eqref{eqn:overall_loss}, but without SupCon (the same uniform sampling strategy as in \cite{xu2021anytime}).
    \item \textbf{GO-OAE-U-SC:} Enhanced \textbf{GO-OAE-U} with SupCon.
    \item \textbf{GO-OAE-I-SC:} Proposed method, \textbf{GO-OAE} with \ac{MPIS} estimator (Algorithm~\ref{alg:go_ovqae}, sampling $m{=}6$ prefixes per minibatch).
\end{itemize}

\subsection{Datasets}\label{ssec:datasets}
We evaluate the proposed \ac{GO-OAE} framework on MiniImageNet \cite{ravi2017optimization,russakovsky2015imagenet}, a widely used fine-grained classification benchmark consisting of 100 classes. We use the public \texttt{timm/mini-imagenet} split via HuggingFace \cite{timm_mini_imagenet_hf}:  $\sim$50k samples for training,  $\sim$10k samples for validation, and $\sim$5k samples for testing. All $100$ classes appear in every split with consistent label indices. Images are resized to $96{\times}96$ (to create $36$ spatial patches of size $16{\times}16$) and normalized to $[0,1]$. Training augmentation includes random crop, horizontal flip, color jitter, rotation, MixUp ($\alpha{=}0.2$), and light blurring and noise corruption. Validation uses center resize and ImageNet normalization only. Top-1 accuracy on the validation set is the primary goal-oriented metric. 

\subsection{Network Architecture and Training Details}\label{ssec:network_architecture}
%
%
We conducted experiments using both \ac{CNN}- and \ac{ViT}-based backbones under the same full-rate budget of $324$ bits per image. Due to space limitations, we focus here on the lighter CNN-based model and leave the backbone discussion, including insights into ViT-based models, to Section~\ref{ssec:backbone}.

{\bf \ac{CNN} backbone.} The encoder maps each image to a $36{\times}6{\times}6$ latent tensor, through four strided convolutional stages (\texttt{cnn\_factor=2}, channels $128{\to}256{\to}512$), a residual block, a final $3{\times}3$ convolution projecting to $36$ channels, and adaptive pooling. The receiver feeds the (possibly truncated) quantized tensor to a ResNet-style classification head: $3{\times}3$ convolutions, batch norm, ReLU, two residual blocks, two stride-$2$ downsamples, average pooling, and a linear classifier, yielding class logits.
 
{\bf \ac{ViT} backbone.} We use a split \texttt{ViT-Small} backbone (resized to $96{\times}96$, yielding $N{=}36$ patch tokens of dimension $384$). The transmitter applies the first $E{=}4$ transformer blocks, followed by a 2-layer perceiver ($36$ queries, $8$ heads) to a latent dimension $D_z{=}128$, then vector-quantized.
The receiver maps quantized latent components back to $384$-dimensional token embeddings, expands to $36$ spatial tokens with a 2-layer perceiver, and runs the remaining $8$ ViT blocks plus the head with separate CLS and positional embeddings. Classification is performed from the final CLS representation.

The \ac{CNN}-based model contains about $8$M, $0.02$M, and $3.3$M parameters in the encoder, VQ codebook, and decoder, respectively. The corresponding parameter counts for the \ac{ViT}-based model are about $11$M, $0.07$M, and $17.9$M.

Unless otherwise specified, both architectures use a trainable \ac{EMA} codebook with size $N{=}512$ ($9$ bits per quantized latent component) and full prefix length $K{=}36$. The per-quantized latent component embedding dimension is $36$ (\ac{CNN}) and $128$ (\ac{ViT}). The \ac{EMA} decay factor is $\gamma{=}0.99$. VQ training employs commitment weight $\beta{=}1.0$ and entropy regularization $\lambda_{\text{ent}}{=}0.02$. 
On SupCon-enabled variants, we add a supervised contrastive term with weight $\lambda_{\text{SC}}{=}3.0$ and temperature $\tau{=}0.1$, applied to pooled pre-quantization features.
For \ac{MPIS}, we draw $m{=}6$ prefix lengths per minibatch with sharpness $\alpha{=}1.0$.

We train the ordered models in three stages:
\begin{itemize}
\item \textbf{Pretrain:} Jointly train the encoder, VQ codebook, and classifier on the full $K{=}36$ prefix ($500$ epochs).
\item \textbf{Ordered:} Fine-tune with random single prefix ({\bf GO-OAE-U-SC}) or multi-prefix ({\bf GO-OAE-I-SC}) to learn an ordered space ($300$ epochs).
\item \textbf{Calibration:} Mix full-prefix ($40\%$) and truncated-prefix ($60\%$) steps to recover full-rate accuracy while preserving low-$k$ performance ($100$ epochs).
\end{itemize}

\subsection{Performance Comparison}\label{ssec:performance_comparison}
\subsubsection{Reconstruction Fidelity vs. Downstream Task Performance}\label{sssec:reconstruction_fidelity}
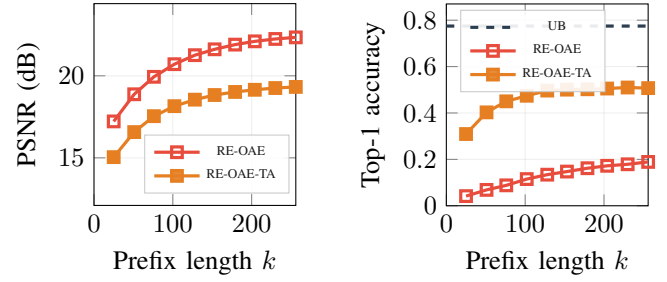
\begin{figure}[t]
  \centering
  \begin{minipage}{0.48\linewidth}
    \centering
    \begin{tikzpicture}
    \begin{axis}[
      width=1\linewidth,
      height=1\linewidth,
      xlabel={Prefix length $k$},
      ylabel={PSNR (dB)},
      grid=major,
      grid style={opacity=0.25},
      xmin=0, xmax=256,
      ymin=12.1, ymax=24.4,
      legend style={at={(0.25,0.35)}, anchor=north west, font=\tiny, fill=white, fill opacity=0.9, draw=gray!45, line width=0.2pt},
    ]
      \addplot[color=RE_OAE, mark=square, mark size=2.0pt, solid, line width=1.2pt, mark options={solid, fill=RE_OAE}] coordinates {(25,17.2211) (51,18.8754) (76,19.9388) (102,20.7086) (128,21.2583) (153,21.6162) (179,21.9021) (204,22.1015) (230,22.2461) (256,22.3435)};
      \addplot[color=RE_OAE_TA, mark=square*, mark size=2.0pt, solid, line width=1.2pt, mark options={solid, fill=RE_OAE_TA}] coordinates {(25,15.0437) (51,16.5702) (76,17.5407) (102,18.1500) (128,18.5471) (153,18.8297) (179,19.0167) (204,19.1483) (230,19.2554) (256,19.3226)};
      \addlegendimage{color=RE_OAE, mark=square, mark size=1.0pt, solid, line width=0.7pt, mark options={solid, fill=RE_OAE}}
      \addlegendentry{RE-OAE}
      \addlegendimage{color=RE_OAE_TA, mark=square*, mark size=1.0pt, solid, line width=0.7pt, mark options={solid, fill=RE_OAE_TA}}
      \addlegendentry{RE-OAE-TA}
    \end{axis}
    \end{tikzpicture}%
     \end{minipage}\hfill
\begin{minipage}{0.48\linewidth}
    \centering
    \begin{tikzpicture}
    \begin{axis}[
      width=1\linewidth,
      height=1\linewidth,
      xlabel={Prefix length $k$},
      ylabel={Top-1 accuracy},
      grid=major,
      grid style={opacity=0.25},
      xmin=0, xmax=256,
      ymin=0.00, ymax=0.87,
      legend style={at={(0.07,0.97)}, anchor=north west, font=\tiny, fill=white, fill opacity=0.9, draw=gray!45, line width=0.2pt},
    ]
    \addplot[color=UB, dashed, line width=1.2pt] coordinates {(0,0.774700) (256,0.774700)};
      \addplot[color=RE_OAE, mark=square, mark size=2.0pt, solid, line width=1.2pt, mark options={solid, fill=RE_OAE}] coordinates {(25,0.0416) (51,0.0684) (76,0.0880) (102,0.1150) (128,0.1340) (153,0.1474) (179,0.1618) (204,0.1722) (230,0.1788) (256,0.1890)};
      \addplot[color=RE_OAE_TA, mark=square*, mark size=2.0pt, solid, line width=1.2pt, mark options={solid, fill=RE_OAE_TA}] coordinates {(25,0.3092) (51,0.4026) (76,0.4502) (102,0.4742) (128,0.4966) (153,0.5006) (179,0.5022) (204,0.5054) (230,0.5098) (256,0.5072)};
      \addlegendimage{color=UB, dashed, line width=1.2pt}
      \addlegendentry{UB}
      \addlegendimage{color=RE_OAE, mark=square, mark size=1.0pt, solid, line width=0.7pt, mark options={solid, fill=RE_OAE}}
      \addlegendentry{RE-OAE}
      \addlegendimage{color=RE_OAE_TA, mark=square*, mark size=1.0pt, solid, line width=0.7pt, mark options={solid, fill=RE_OAE_TA}}
      \addlegendentry{RE-OAE-TA}
    \end{axis}
    \end{tikzpicture}%
    \end{minipage}
  \caption{Reconstruction fidelity vs.\ downstream task performance for the reconstruct-then-classify pipelines \textbf{RE-OAE} and \textbf{RE-OAE-TA} ($K{=}256$).}
  \label{fig:recon_fidelity_vs_accuracy}
\end{figure}
\input{Figures/fig_recon_samples_labels}

Our reconstruction-based evaluation in Figure~\ref{fig:recon_fidelity_vs_accuracy} and~\ref{fig:recon_samples} shows that strong reconstruction performance does not translate into reliable downstream classification. Even at a substantially larger transmission budget ($K{=}256$ already offers better performance than $K{=}36$ in Figure~\ref{fig:perf_comp}), {\bf RE-OAE} attains satisfying \ac{PSNR} and visually plausible reconstructions, yet top-1 accuracy on decoded images remains far below the original-image bound ($18.9\%$ vs.\ $77.5\%$). Adding task-aware supervision in \textbf{RE-OAE-TA} substantially improves performance, but still falls well short of the upper bound. These results indicate that a modular ``compress first, infer later'' pipeline can fail even when reconstruction error is small. Goal-oriented design is therefore necessary to align the transmitted representation with task-relevant information.
\subsubsection{Performance Comparison}\label{sssec:OAE_vs_GO_OAE}
\begin{figure}[t]
  \centering
\begin{tikzpicture}
\begin{axis}[
  width=0.48\textwidth,
  height=0.32\textwidth,
  xlabel={Prefix length $k$},
  ylabel={Top-1 Accuracy},
  grid=major,
  grid style={opacity=0.25},
  xmin=3,
  ymin=0,
  ymax=0.8,
  xmax=36,
  legend style={at={(0.6,0.7)}, anchor=north west, font=\tiny, fill=white, fill opacity=0.9, draw=gray!45, line width=0.2pt},
]
\addplot[color=fig1E1, dashed, line width=1.2pt, forget plot] coordinates {(0,0.774700) (36,0.774700)};
  \addplot[color=fig1E2, mark=square, mark size=2.0pt, line width=1.2pt, forget plot] coordinates {(3,0.012700) (7,0.014500) (10,0.018900) (14,0.024400) (18,0.026900) (21,0.028400) (25,0.034000) (28,0.038700) (32,0.042000) (36,0.047100)};
  \addplot[color=fig1E3, mark=square*, mark size=2.0pt, solid, line width=1.2pt, mark options={solid, fill=fig1E3}, forget plot] coordinates {(3,0.017200) (7,0.066600) (10,0.170800) (14,0.308500) (18,0.408000) (21,0.446300) (25,0.486900) (28,0.513000) (32,0.538400) (36,0.553700)};
  \addplot[color=fig1E40, mark=diamond, mark size=2.0pt, densely dotted, line width=1.2pt, mark options={solid}, forget plot] coordinates {(3,0.012400) (7,0.016400) (10,0.015900) (14,0.026100) (18,0.051500) (21,0.063900) (25,0.106600) (28,0.151900) (32,0.208800) (36,0.262200)};
  \addplot[color=fig1E40+supcon, mark=diamond*, mark size=2.0pt, densely dotted, line width=1.2pt, mark options={solid, fill=fig1E40+supcon}, forget plot] coordinates {(3,0.037600) (7,0.100800) (10,0.176700) (14,0.314900) (18,0.452400) (21,0.499200) (25,0.554900) (28,0.573100) (32,0.598500) (36,0.637600)};
  \addplot[color=fig1E4a, mark=o, mark size=2.0pt, solid, line width=1.2pt, forget plot] coordinates {(3,0.049300) (7,0.233800) (10,0.382300) (14,0.507300) (18,0.577400) (21,0.621200) (25,0.647600) (28,0.677000) (32,0.697700) (36,0.707400)};
  \addplot[color=fig1E4a+supcon, mark=*, mark size=2.0pt, solid, line width=1.2pt, mark options={solid, fill=fig1E4a+supcon}, forget plot] coordinates {(3,0.131400) (7,0.479800) (10,0.598100) (14,0.665400) (18,0.693800) (21,0.706100) (25,0.719000) (28,0.725700) (32,0.730100) (36,0.731400)};
  \addplot[color=fig1E4b, mark=asterisk, mark size=2.8pt, solid, line width=1.5pt, mark options={solid, fill=fig1E4b}, forget plot] coordinates {(3,0.270000) (7,0.556300) (10,0.643300) (14,0.689200) (18,0.705500) (21,0.718200) (25,0.728600) (28,0.733100) (32,0.741300) (36,0.744300)};
  \addlegendimage{color=fig1E1, dashed, line width=0.7pt, forget plot}
  \addlegendentry{UB}
  \addlegendimage{color=fig1E2, mark=square, mark size=1.0pt, line width=0.7pt}
  \addlegendentry{RE-OAE}
  \addlegendimage{color=fig1E3, mark=square*, mark size=1.0pt, solid, line width=0.7pt, mark options={solid, fill=fig1E3}}
  \addlegendentry{RE-OAE-TA}
  \addlegendimage{color=fig1E40, mark=diamond, mark options={solid}, mark size=1.0pt, densely dotted, line width=0.7pt}
  \addlegendentry{GO-VQAE}
  \addlegendimage{color=fig1E40+supcon, mark=diamond*, mark size=1.0pt, densely dotted, line width=0.7pt, mark options={solid, fill=fig1E40+supcon}}
  \addlegendentry{GO-VQAE-SC}
  \addlegendimage{color=fig1E4a, mark=o, mark size=1.0pt, solid, line width=0.7pt, mark options={solid, fill=fig1E4a}}
  \addlegendentry{GO-OAE-U}
  \addlegendimage{color=fig1E4a+supcon, mark=*, mark size=1.0pt, solid, line width=0.7pt, mark options={solid, fill=fig1E4a+supcon}}
  \addlegendentry{GO-OAE-U-SC}
  \addlegendimage{color=fig1E4b, mark=asterisk, mark size=1.0pt, solid, line width=0.7pt, mark options={solid, fill=fig1E4b}}
  \addlegendentry{GO-OAE-I-SC}
\end{axis}
\end{tikzpicture}
\vspace{-2ex}
  \caption{Task performance comparison (\ac{CNN} backbone, $K=36$).}
  \label{fig:perf_comp}
\end{figure}
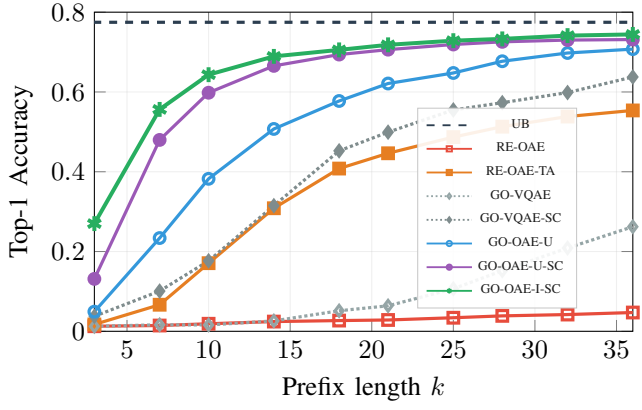
Figure~\ref{fig:perf_comp} reports mean top-1 accuracy (three runs per scenario) as a function of selected prefix length ($k{\le}K{=}36$) for all eight scenarios applied to \ac{CNN} backbone. Our main observations are as follows.
First, \textbf{RE-OAE} and \textbf{RE-OAE-TA} remain weak across all prefix lengths: at the full budget they reach only $4.7\%$ and $55.4\%$, respectively, versus $77.5\%$ for the original-image bound (\textbf{UB}), confirming that a modular ``compress first, infer later'' pipeline can fail despite low reconstruction error. Task performance is also much lower than the results with higher prefix budget $K{=}256$ (Figure~\ref{fig:recon_fidelity_vs_accuracy} and \ref{fig:recon_samples}).
Second, end-to-end goal-oriented training improves accuracy substantially. \textbf{GO-VQAE} and \textbf{GO-VQAE-SC} jointly optimize the encoder, codebook, and task head with maximum prefix length (complete sequence, no truncation) without ordering. They serve as unordered references and show that direct task optimization outperforms reconstruction-first design. However, without ordered training, the budget $K{=}36$ is too tight for them to obtain sufficient task-relevant information. 
Third, semantic preservation in quantized latent space yields large gains in both unordered and ordered settings: \textbf{GO-VQAE-SC} improves over \textbf{GO-VQAE} by $+37.5$ percentage points at $k=36$ ($63.8\%$ vs.\ $26.2\%$), and \textbf{GO-OAE-U-SC} improves over \textbf{GO-OAE-U} by $+2.4$ points ($73.1\%$ vs.\ $70.7\%$); the effect is even larger at low budgets.
%
Fourth, ordered training (\textbf{GO-OAE-U}, \textbf{GO-OAE-U-SC}, \textbf{GO-OAE-I-SC}) strongly boosts accuracy at small prefixes, enabling more robust performance under tight communication budgets.
Finally, \textbf{GO-OAE-I-SC} approaches \textbf{UB} with limited prefix budget: at $k{=}36$ it reaches accuracy $74.4\%$ ($96.1\%$ of \textbf{UB}) using only 324 transmitted bits ($683\times$ compression vs.\ uncompressed RGB). Also, with only $k{=}25$ code indices (225 bits) it already attains $72.9\%$ accuracy ($94.1\%$ of \textbf{UB}, $980\times$ compression), highlighting the efficiency of goal-oriented ordered representations under communication constraints.

\subsubsection{Semantic Preservation in Codebook Embeddings - Effect of Supervision}\label{sssec:semantic_preservation}
Figure~\ref{fig:perf_comp} already showed that semantic preservation in quantized latent space, enforced by a supervised contrastive regularizer, yields large gains in both unordered and ordered settings. Figure~\ref{fig:semantic_preserve} (the same \ac{CNN} backbone) complements this result in two ways: (i)~\textbf{GO-OAE-I-SC} maintains stronger semantic structure in the quantized latent space than \textbf{GO-VQAE}, and (ii)~our \ac{EMA} codebook update preserves semantic information through quantization (pre-VQ vs. post-VQ).



\begin{figure}[t]
  \centering
\begin{tikzpicture}
  \node[inner sep=0] (plots) {\includegraphics[width=0.95\linewidth]{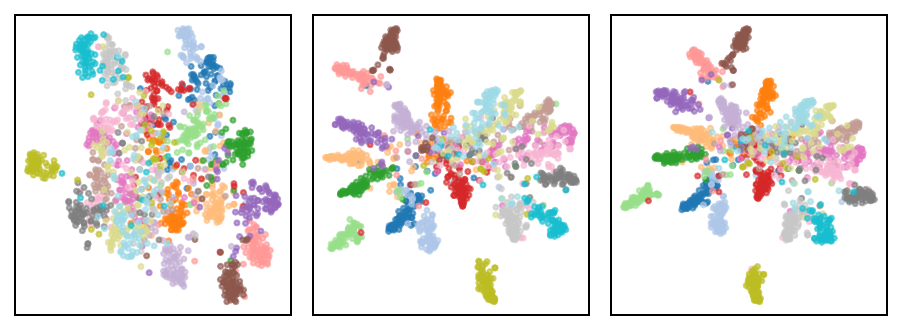}};
\end{tikzpicture}

\vspace{0.25em}
\noindent
  \begin{minipage}{0.3\linewidth}\centering
    {\small GO-VQAE}\\
    {\small (pre-VQ)}
  \end{minipage}
  \begin{minipage}{0.3\linewidth}\centering
    {\small GO-OAE-I-SC}\\
    {\small (pre-VQ)}
  \end{minipage}
  \begin{minipage}{0.3\linewidth}\centering
    {\small GO-OAE-I-SC}\\
    {\small (post-VQ)}
  \end{minipage}
    \caption{T-SNE of latent components colored by $20$ classes.}
  \label{fig:semantic_preserve}
\end{figure}

\subsubsection{Ordered Latent Space - Effect of Sampling Strategies}   \label{sssec:sampling_strategies}
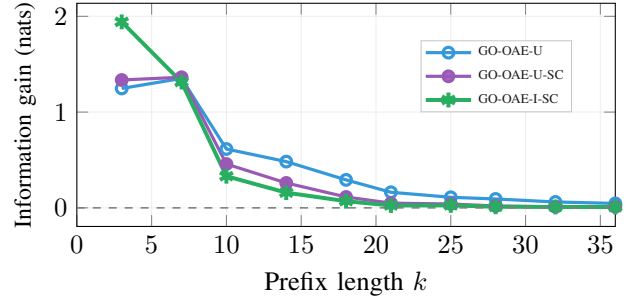
\begin{figure}[t]
  \centering
\begin{tikzpicture}
\begin{axis}[
  width=0.48\textwidth,
  height=0.25\textwidth,
  xlabel={Prefix length $k$},
  ylabel={Information gain (nats)},
  ylabel style={font=\small},
  grid=major,
  grid style={opacity=0.25},
  xmin=0,
  xmax=36,
  legend style={
    at={(axis cs:23,1)},
    anchor=south west,
    font=\tiny,
    fill=white,
    fill opacity=0.9,
    draw=gray!45,
    line width=0.2pt,
    inner xsep=1pt,
    inner ysep=0pt,
    legend cell align=left,
  },
]
  \addplot[color=fig1E4a, mark=o, mark size=2.0pt, solid, line width=1.2pt, mark options={solid, fill=fig1E4a}] coordinates {(3,1.246721) (7,1.351610) (10,0.614524) (14,0.482953) (18,0.291934) (21,0.162742) (25,0.110243) (28,0.092749) (32,0.061399) (36,0.046276)};
  \addplot[color=fig1E4a+supcon, mark=*, mark size=2.0pt, solid, line width=1.2pt, mark options={solid, fill=fig1E4a+supcon}] coordinates {(3,1.335080) (7,1.363568) (10,0.457572) (14,0.259932) (18,0.113599) (21,0.049092) (25,0.040387) (28,0.019399) (32,0.010690) (36,0.011505)};
  \addplot[color=fig1E4b, mark=asterisk, mark size=2.8pt, solid, line width=1.5pt, mark options={solid, fill=fig1E4b}] coordinates {(3,1.941881) (7,1.313617) (10,0.330486) (14,0.158861) (18,0.069635) (21,0.024782) (25,0.025868) (28,0.011468) (32,0.009341) (36,0.007984)};
  \addplot[gray, dashed, line width=0.6pt] coordinates {(0,0) (47,0)};
  \addlegendimage{color=fig1E4a, mark=o, mark size=1.0pt, solid, line width=0.7pt, mark options={solid, fill=fig1E4a}}
  \addlegendentry{GO-OAE-U}
  \addlegendimage{color=fig1E4a+supcon, mark=*, mark size=1.0pt, solid, line width=0.7pt, mark options={solid, fill=fig1E4a+supcon}}
  \addlegendentry{GO-OAE-U-SC}
  \addlegendimage{color=fig1E4b, mark=asterisk, mark size=1.0pt, solid, line width=0.7pt, mark options={solid, fill=fig1E4b}}
  \addlegendentry{GO-OAE-I-SC}
\end{axis}
\end{tikzpicture}
\vspace{-2ex}
  \caption{Information gain comparison (\ac{CNN} backbone, $K=36$).}
  \label{fig:delta_CE}
\end{figure}
\begin{figure}[!ht]
  \centering
\begin{tikzpicture}
\begin{axis}[
  width=0.48\textwidth,
  height=0.25\textwidth,
  xlabel={Prefix length $k$},
  ylabel={Top-1 Accuracy},
  ylabel style={font=\small},
  grid=major,
  grid style={opacity=0.25},
  xmin=0,
  xmax=36,
  legend style={
    at={(axis cs:20,0.3)},
    anchor=south west,
    font=\tiny,
    fill=white,
    fill opacity=0.9,
    draw=gray!45,
    line width=0.2pt,
    inner xsep=1pt,
    inner ysep=0pt,
    legend cell align=left,
  },
]
\addplot[color=fig1E4b, dashed, line width=1.2pt, forget plot] coordinates {(0,0.774700) (36,0.774700)};
\addplot[color=ViT, dotted, line width=1.2pt, forget plot] coordinates {(0,0.843700) (36,0.843700)};
\addplot[color=fig1E4b, mark=asterisk, mark size=2.8pt, solid, line width=1.2pt, mark options={solid, fill=fig1E4b}, forget plot] coordinates {(3,0.270000) (7,0.556300) (10,0.643300) (14,0.689200) (18,0.705500) (21,0.718200) (25,0.728600) (28,0.733100) (32,0.741300) (36,0.744300)};
\addplot[color=ViT, mark=triangle*, mark size=2.8pt, solid, line width=1.2pt, mark options={solid, fill=ViT}, forget plot] coordinates {(3,0.752700) (7,0.784900) (10,0.795500) (14,0.796000) (18,0.796300) (21,0.796300) (25,0.797100) (28,0.797300) (32,0.797200) (36,0.797200)};
\addplot[color=ViT, mark=triangle, mark size=2.8pt, dashed, line width=1.2pt, mark options={solid, fill=ViT}, forget plot] coordinates {(3,0.634100) (7,0.694000) (10,0.718800) (14,0.725400) (18,0.728700) (21,0.731000) (25,0.734800) (28,0.734900) (32,0.735000) (36,0.735000)};

  \addlegendimage{color=fig1E4b, dashed, line width=0.7pt, forget plot}
  \addlegendentry{UB (CNN)}
    \addlegendimage{color=ViT, dotted, line width=0.7pt, forget plot}
  \addlegendentry{UB (ViT-S)}
  \addlegendimage{color=fig1E4b, mark=asterisk, mark size=1.0pt, solid, line width=0.7pt, mark options={solid, fill=fig1E4b}}
  \addlegendentry{GO-OAE-I-SC (CNN)}
    \addlegendimage{color=ViT, mark=triangle*, mark size=1.0pt, solid, line width=0.7pt, mark options={solid, fill=ViT}}
  \addlegendentry{GO-OAE-I-SC (ViT-S, E=4)}
      \addlegendimage{color=ViT, mark=triangle, mark size=1.0pt, dashed, line width=0.7pt, mark options={solid, fill=ViT}}
  \addlegendentry{GO-OAE-I-SC (ViT-S, E=1)}
\end{axis}
\end{tikzpicture}
\vspace{-2ex}
  \caption{Backbone comparison.}
  \label{fig:compare_cnn_vit}
\end{figure}
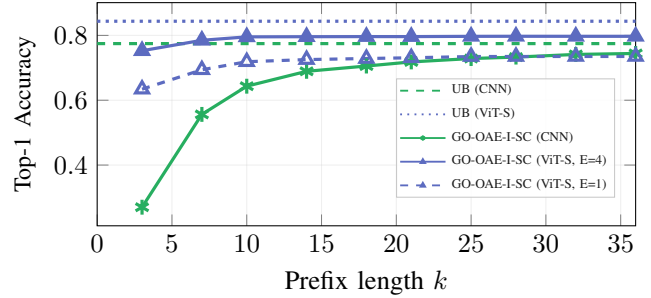
Figure \ref{fig:delta_CE} compares the block information gain of  \textbf{GO-OAE-U}, \textbf{GO-OAE-U-SC} and \textbf{GO-OAE-I-SC}. We partition the ordered codes into $10$ blocks, where block $p$ ends at prefix length $k_p$ corresponding to $p\in\{10\%, 20\%,\ldots, 100\%\}$ of the full budget. It reports the blockwise empirical information gain $\Delta\widehat{I}_p = \mathrm{CE}(\widehat{\vy}{\mid}\vz_{\le k_{p-1}}) - \mathrm{CE}(\widehat{\vy}{\mid}\vz_{\le k_p})$.
The curve of \textbf{GO-OAE-I-SC} decays sharply across prefix length, confirming the ordered autoencoding requirement of monotonically diminishing conditional mutual information, consistent with the ordering predicted by Corollary \ref{corollary:GO_OAE_ordering}.

\subsection{Backbone Model Discussion}\label{ssec:backbone}
Our experiments in Figure \ref{fig:compare_cnn_vit} show that \ac{GO-OAE} interacts differently with \ac{CNN} and \ac{ViT} backbones. CNN's feature specialization supports complementary latent information and clear progressive refinement. In contrast, standard \acp{ViT} are optimized for early global contextualization, which conflicts with ordered representation learning: task-relevant information becomes rapidly mixed and redundantly distributed across tokens, allowing the first few latent components to capture most of the utility and yielding nearly flat prefix-performance curves. Using fewer transformer blocks before quantization ($E{=}1$ instead of $E{=}4$) spreads task utility across more latent components, yielding progressive gains over longer prefix. However, the reduced depth at transmitter also limits task-relevant information extraction before the bottleneck, resulting in lower converged accuracy. Achieving progressive refinement with transformer backbones may therefore require architectural redesigns that preserve decomposable intermediate features and defer global integration until after the communication bottleneck. This represents an important direction for future work.

\section{Conclusions}\label{sec:conclusions}
We introduced \ac{GO-OAE}, a goal-oriented progressive communication framework robust to unpredictable bandwidth that learns ordered quantized representations, enabling flexible task inference rather than source reconstruction. We showed that the goal-oriented distortion objective admits a weighted conditional information bottleneck interpretation, and introduced \ac{MPIS} to cut training cost from $\mathcal{O}(BK)$ to $\mathcal{O}(Bm)$, with $m{\ll}K$, and with variance scaling as $1/m$. On MiniImageNet, our approach attains ${\geq}96\%$ of the lossless bound at 324 bits/image ($683\times$ compression) and ${\geq}94\%$ with only $70\%$ of the quantized latent components, outperforming reconstruction-based codecs and \acp{OAE} trained with uniform prefix sampling at all budgets. Semantic preservation regularization and ordered importance sampling are essential when the communication budget is tight. Future work will extend \ac{GO-OAE} to drone surveillance and robotic vision under realistic communication, compute, and energy constraints, while developing transformer architectures better suited to progressive refinement.
\appendix
\subsection{Proof of Theorem~\ref{theorem:GO_OAE_objective}}\label{ssec:proof_GO_OAE_objective}
\begin{proof}
    Assuming Bayes-optimal decoders, for which $\E{-\log p(\vy|\vz_{\leq i})}{=}H(\vy|\vz_{\leq i})$, where $H(\cdot|\cdot)$ denotes the conditional entropy. Then, we have
     \vspace{-1ex}
    \begin{equation} 
    \Ls_D = \sum_{i=1}^K H(\vy|\vz_{\leq i}) =K H(\vy) - \sum_{i=1}^K I(\vy; \vz_{\leq i})\nonumber
    \end{equation}
    Thus, minimizing $\Ls_D$ is equivalent to maximizing $\sum_{i=1}^K I(\vy; \vz_{\leq i})$. Based on the chain rule, we have 
    \vspace{-1ex}
    $$I(\vy; \vz_{\leq i}) = \sum_{j=1}^i I(\vy; \vz_j|\vz_{< j}),    \vspace{-1ex}$$
    which leads to
    \begin{align}
        \sum_{i=1}^K I(\vy; \vz_{\leq i}) & = \sum_{i=1}^K \sum_{j=1}^i I(\vy; \vz_j|\vz_{< j})\nonumber\\
        & = \sum_{j=1}^K (K-j+1) I(\vy; \vz_j|\vz_{< j})\nonumber.    
    \end{align}
    \vspace{-1ex}
    The loss becomes
    \begin{equation}
    \vspace{-1ex}
        \Ls_D = K H(\vy) - \sum_{i=1}^K w_i I(\vy; \vz_i|\vz_{< i}),\nonumber
    \end{equation}
    where the weights are defined as $w_i = K-i+1$, which means information contributed by earlier latent components is rewarded more strongly.
\end{proof}

\subsection{Proof of Corollary~\ref{corollary:GO_OAE_ordering}}\label{ssec:proof_GO_OAE_ordering}
\begin{proof}
    From Theorem \ref{theorem:GO_OAE_objective}, we know minimizing \eqref{eqn:goal_oriented_distortion} is equivalent, up to an additive constant, to maximizing \eqref{eqn:GO_OAE_IB} with non-increasing weights $w_i{=}K-i+1$, and therefore $w_1{>}w_2{>}\ldots{>}w_K{>}0$. Define the incremental task-relevant information gain as $\Delta_i = I(\vy; \vz_i|\vz_{< i})$ (for brevity we abuse the notation and still use $\Delta_i$ for goal-oriented information gain). Suppose, for contradiction, that an optimal representation contains an adjacent inversion, i.e., $\Delta_i{<}\Delta_{i+1}$ for some $i$. After swapping their order, the contribution becomes $w_i \Delta_{i+1}{+}w_{i+1} \Delta_i$, which is strictly larger than the original contribution $w_i \Delta_i + w_{i+1} \Delta_{i+1}$. Thus, swapping the two latent components would strictly increase the objective. Since the model class is closed under permutations and the information contributions are invariant under permutation, the swapped representation remains admissible. This contradicts the global optimality of the original representation. Therefore, no adjacent inversion can exist at optimum, and hence $\Delta_1{\geq}\Delta_2 {\geq}\ldots{\geq}\Delta_K$.
\end{proof}

\subsection{Proof of Corollary~\ref{corollary:unbiased_estimator}}\label{ssec:proof_unbiased_estimator}
\begin{proof}
In below, we prove the unbiasedness, variance, and unbiased gradient estimation properties of the proposed estimator $\widehat{\Ls}_D$, respectively.

{\bf Unbiasedness:} By linearity of expectation, we have
$$\E{\widehat{\Ls}_D} = \E{\frac{1}{m}\sum_{n=1}^m\EE_{i_n\sim q}\left[\frac{L_{i_n}}{q(i_n)}\right]}.$$
Since $i_n$ is sampled from $q(i)$, we have 
\begin{equation}
\EE_{i_n\sim q}\left[\frac{L_{i_n}}{q(i_n)}\right] = \sum_{i=1}^K q(i) \frac{L_i}{q(i)} = \sum_{i=1}^K L_i = \Ls_D.
\label{eqn:Ls_D}
\end{equation}
Thus, $\E{\widehat{\Ls}_D} = \Ls_D$, which proves the unbiasedness. 

{\bf Variance:} Let us define $X_n:=L_{i_n}/q(i_n)$ for $n=1, \ldots, m$. Since $i_n$ are sampled independently, and all $X_1, \ldots, X_m$ are i.i.d. and have identical variance, we have
$$\Var\left[\widehat{\Ls}_D\right] = \frac{1}{m^2}\sum_{n=1}^m\Var[X_n]=\frac{1}{m}\Var\left[\frac{L_i}{q(i)}\right].$$
Given the variance definition, we have
\begin{align}  
\Var\left[\frac{L_i}{q(i)}\right] & = \EE_{i\sim q}\left[\left(\frac{L_i}{q(i)}\right)^2\right] - \left(\EE_{i\sim q}\left[\frac{L_i}{q(i)}\right]\right)^2 \nonumber\\
& \overset{\eqref{eqn:Ls_D}}{=} \sum_{i=1}^K q(i) \frac{L_i^2}{q(i)^2}- \Ls_D^2 = \sum_{i=1}^K \frac{L_i^2}{q(i)} - \Ls_D^2.
\end{align}
Thus, we have $\Var\left[\widehat{\Ls}_D\right] = (1/m) \left(\sum_{i=1}^K L_i^2/q(i) -  \Ls_D^2\right)$.

\vspace{1ex}
{\bf Unbiased Gradient Estimation:} With the proven unbiasedness, and under standard regularity conditions allowing interchange of gradient and expectation, we have
$$\nabla_{\theta, \phi}\Ls_D(\theta, \phi)=\nabla_{\theta, \phi}\EE[\widehat{\Ls}_D(\theta, \phi)] = \EE[\nabla_{\theta, \phi}\widehat{\Ls}_D(\theta, \phi)].$$         
\end{proof}
\bibliographystyle{IEEEtran}
\bibliography{references}

\begin{thebibliography}{10}
\providecommand{\url}[1]{#1}
\csname url@samestyle\endcsname
\providecommand{\newblock}{\relax}
\providecommand{\bibinfo}[2]{#2}
\providecommand{\BIBentrySTDinterwordspacing}{\spaceskip=0pt\relax}
\providecommand{\BIBentryALTinterwordstretchfactor}{4}
\providecommand{\BIBentryALTinterwordspacing}{\spaceskip=\fontdimen2\font plus
\BIBentryALTinterwordstretchfactor\fontdimen3\font minus \fontdimen4\font\relax}
\providecommand{\BIBforeignlanguage}[2]{{%
\expandafter\ifx\csname l@#1\endcsname\relax
\typeout{** WARNING: IEEEtran.bst: No hyphenation pattern has been}%
\typeout{** loaded for the language `#1'. Using the pattern for}%
\typeout{** the default language instead.}%
\else
\language=\csname l@#1\endcsname
\fi
#2}}
\providecommand{\BIBdecl}{\relax}
\BIBdecl

\bibitem{goldreich2012theory}
O.~Goldreich, B.~Juba, and M.~Sudan, ``A theory of goal-oriented communication,'' \emph{Journal of the ACM (JACM)}, vol.~59, no.~2, pp. 1--65, 2012.

\bibitem{mentzer2020high}
F.~Mentzer, G.~D. Toderici, M.~Tschannen, and E.~Agustsson, ``High-fidelity generative image compression,'' \emph{Advances in neural information processing systems (NeurIPS)}, vol.~33, pp. 11\,913--11\,924, 2020.

\bibitem{zhang2025progressive}
G.~Zhang and et~al., ``Progressive learned image transmission for semantic communication using hierarchical {VAE},'' \emph{IEEE Transactions on Cognitive Communications and Networking}, 2025.

\bibitem{diederik2019introduction}
D.~P. Kingma and M.~Welling, ``An introduction to variational autoencoders,'' \emph{Foundations and Trends{\textregistered} in Machine Learning}, vol.~12, no.~4, pp. 307--392, 2019.

\bibitem{equitz1991successive}
W.~H. Equitz and T.~M. Cover, ``Successive refinement of information,'' \emph{IEEE Transactions on information theory}, vol.~37, no.~2, pp. 269--275, 1991.

\bibitem{xu2021anytime}
Y.~Xu, Y.~Song, S.~Garg, L.~Gong, R.~Shu, A.~Grover, and S.~Ermon, ``Anytime sampling for autoregressive models via ordered autoencoding,'' in \emph{International Conference on Learning Representations (ICLR)}, 2021.

\bibitem{liang2025semantic}
H.~Liang, C.~Dong, W.~An, Z.~Bao, X.~Xu, and R.~Meng, ``Semantic importance-aware communication over {MIMO} fading channels,'' \emph{IEEE Internet of Things Journal}, 2025.

\bibitem{ravi2017optimization}
S.~Ravi and H.~Larochelle, ``Optimization as a model for few-shot learning,'' in \emph{International Conference on Learning Representations (ICLR)}, 2017.

\bibitem{russakovsky2015imagenet}
O.~Russakovsky, J.~Deng, H.~Su, J.~Krause, S.~Satheesh, S.~Ma, Z.~Huang, A.~Karpathy, A.~Khosla, M.~Bernstein \emph{et~al.}, ``Imagenet large scale visual recognition challenge,'' \emph{International Journal of Computer Vision (IJCV)}, vol. 115, no.~3, pp. 211--252, 2015.

\bibitem{qin2021semantic}
Z.~Qin, X.~Tao, J.~Lu, W.~Tong, and G.~Y. Li, ``Semantic communications: {P}rinciples and challenges,'' \emph{arXiv preprint arXiv:2201.01389}, 2021.

\bibitem{gunduzsemantic2023}
D.~Gündüz, Z.~Qin, I.~E. Aguerri, H.~S. Dhillon, Z.~Yang, A.~Yener, K.~K. Wong, and C.-B. Chae, ``Beyond transmitting bits: Context, semantics, and task-oriented communications,'' \emph{IEEE Journal on Selected Areas in Communications}, vol.~41, no.~1, pp. 5--41, 2023.

\bibitem{getusemantic2024}
T.~M. Getu, G.~Kaddoum, and M.~Bennis, ``A survey on goal-oriented semantic communication: Techniques, challenges, and future directions,'' \emph{IEEE Access}, vol.~12, pp. 51\,223--51\,274, 2024.

\bibitem{luo_2022}
X.~Luo, H.-H. Chen, and Q.~Guo, ``Semantic communications: Overview, open issues, and future research directions,'' \emph{IEEE Wireless Communications}, vol.~29, no.~1, pp. 210--219, 2022.

\bibitem{liao2024adasem}
Q.~Liao and T.-Y. Tung, ``Adasem: Adaptive goal-oriented semantic communications for end-to-end camera relocalization,'' in \emph{IEEE INFOCOM 2024-IEEE Conference on Computer Communications}.\hskip 1em plus 0.5em minus 0.4em\relax IEEE, 2024, pp. 1111--1120.

\bibitem{li2026goal}
N.~Li, Y.~Deng, and D.~Niyato, ``Goal-oriented semantic communication for wireless video transmission via generative ai,'' \emph{IEEE Transactions on Wireless Communications}, 2026.

\bibitem{wanasekara2025sc}
S.~H. Wanasekara, V.-D. Nguyen, K.-S. Wong, M.-D. Nguyen, S.~Chatzinotas, and O.~A. Dobre, ``Sc-gir: Goal-oriented semantic communication via invariant representation learning for image transmission,'' \emph{IEEE Transactions on Mobile Computing}, 2025.

\bibitem{lyu2024semantic}
Z.~Lyu, G.~Zhu, J.~Xu, B.~Ai, and S.~Cui, ``Semantic communications for image recovery and classification via deep joint source and channel coding,'' \emph{IEEE Transactions on Wireless Communications}, vol.~23, no.~8, pp. 8388--8404, 2024.

\bibitem{li2024autoregressive}
T.~Li, Y.~Tian, H.~Li, M.~Deng, and K.~He, ``Autoregressive image generation without vector quantization,'' \emph{Advances in Neural Information Processing Systems (NeurIPS)}, vol.~37, pp. 56\,424--56\,445, 2024.

\bibitem{kusupati2022matryoshka}
A.~Kusupati and et~al., ``Matryoshka representation learning,'' \emph{Advances in Neural Information Processing Systems (NeurIPS)}, vol.~35, pp. 30\,233--30\,249, 2022.

\bibitem{van2017neural}
A.~Van Den~Oord, O.~Vinyals \emph{et~al.}, ``Neural discrete representation learning,'' \emph{Advances in neural information processing systems (NeurIPS)}, vol.~30, 2017.

\bibitem{grattafiori2024llama}
A.~Grattafiori, A.~Dubey, A.~Jauhri, A.~Pandey, A.~Kadian, A.~Al-Dahle, A.~Letman, A.~Mathur, A.~Schelten, A.~Vaughan \emph{et~al.}, ``The {L}lama 3 herd of models,'' \emph{arXiv preprint arXiv:2407.21783}, 2024.

\bibitem{lancucki2020robust}
A.~{\L}a{\'n}cucki, J.~Chorowski, G.~Sanchez, R.~Marxer, N.~Chen, H.~J. Dolfing, S.~Khurana, T.~Alum{\"a}e, and A.~Laurent, ``Robust training of vector quantized bottleneck models,'' in \emph{2020 International Joint Conference on Neural Networks (IJCNN)}.\hskip 1em plus 0.5em minus 0.4em\relax IEEE, 2020, pp. 1--7.

\bibitem{khosla2020supervised}
P.~Khosla, P.~Teterwak, C.~Wang, A.~Sarna, Y.~Tian, P.~Isola, A.~Maschinot, C.~Liu, and D.~Krishnan, ``Supervised contrastive learning,'' \emph{Advances in neural information processing systems (NeurIPS)}, vol.~33, pp. 18\,661--18\,673, 2020.

\bibitem{timm_mini_imagenet_hf}
R.~Wightman, ``timm/mini-imagenet dataset,'' \url{https://huggingface.co/datasets/timm/mini-imagenet}, 2021, accessed: 2026-07-30.

\end{thebibliography}


\end{document}